\documentclass[12pt]{article}
\usepackage{natbib}

\usepackage{amsmath}
\usepackage{graphicx}
\usepackage{enumerate}
\usepackage{url}

\newcommand{\blind}{1}

\usepackage{graphicx}
\usepackage[T1]{fontenc}
\usepackage[utf8]{inputenc}
\usepackage[english]{babel}
\usepackage{amsmath,amsthm,amssymb,amsfonts}
\usepackage{bm}

\usepackage{booktabs} 
\usepackage{threeparttable}
\usepackage{makecell}
\usepackage{multirow}

\usepackage{lmodern}
\usepackage{subcaption}
\usepackage{geometry}
 \usepackage[nice]{nicefrac}
\usepackage{mathabx}
 
\usepackage[colorinlistoftodos,textwidth=2.1cm]{todonotes}

\usepackage{tcolorbox}
\usepackage{pifont}
\definecolor{mydarkgreen}{RGB}{39,130,67}
\definecolor{mydarkred}{RGB}{192,25,25}
\definecolor{mydarkblue}{RGB}{0,0,140}

 \usepackage{empheq}
 \usepackage{color}
\definecolor{darkgreen}{rgb}{0.00,0.5,0.00}
\usepackage[pagebackref=true,backref=true]{hyperref}
\hypersetup{
	colorlinks=true,        
	linkcolor=blue,         
	citecolor=mydarkblue,         
	urlcolor=blue           
}

\usepackage{tcolorbox}
\usepackage{framed}
\usepackage{algorithm}
\usepackage{algpseudocode}
\usepackage{enumerate}
\usepackage{cleveref}
\newtheorem{theorem}{Theorem}
\newtheorem{assumption}{Assumption}
\newtheorem{definition}{Definition}
\newtheorem{lemma}{Lemma}
\newtheorem{proposition}{Proposition}

\crefname{assumption}{assumption}{assumptions}

\theoremstyle{definition}
\newtheorem{remark}{Remark}
\newtheorem{example}{Example}

\newcommand{\WR}{\mathrm{WR}}

\usepackage{sidecap}

\let\lll\ll
\renewcommand{\ll}{\mathbf{l}}

\newcommand{\E}{\mathbb{E}}

\newcommand{\indep}{\perp \!\!\! \perp}

\newcommand{\cC}{\mathcal{C}}
\newcommand{\cS}{\mathcal{S}}
\newcommand{\cY}{\mathcal{Y}}
\newcommand{\cZ}{\mathcal{Z}}

\newcommand{\cR}{\mathcal{R}}

\newcommand{\R}{\mathbb R}

\newcommand{\cE}{\mathcal{E}}

\DeclareMathOperator{\argmin}{argmin}

\def\<#1,#2>{\langle #1,#2\rangle}

\usepackage{dsfont}

\renewcommand{\leq}{\leqslant}
\renewcommand{\geq}{\geqslant}
\renewcommand{\le}{\leqslant}

\def\<{\langle}
\def\>{\rangle}

\def\eps{\varepsilon}
\def\var{{\rm var\,}}

\newcommand{\esp}[1]{\mathbb{E}\left[#1\right]}

\newcommand{\NRM}[1]{{{\left\| #1\right\|}}} 
\newcommand{\set}[1]{{{\left\{ #1\right\}}}} 
\newcommand{\proba}[1]{\mathbb{P}\left(#1\right)}

\renewcommand{\P}{\mathbb{P}}

\newcommand{\cB}{\mathcal{B}}
\newcommand{\cN}{\mathcal{N}}
\newcommand{\cX}{\mathcal{X}}
\newcommand{\cW}{\mathcal{W}}

\newcommand{\dd}{{\rm d}}
\newcommand{\cF}{\mathcal{F}}
\newcommand{\cP}{\mathcal{P}}
\newcommand{\cH}{\mathcal{H}}

\newcommand{\cD}{\mathcal{D}}
\newcommand{\one}{\mathds{1}}

\usepackage{amsfonts}
\usepackage{amsthm}
\usepackage{amsmath}
\usepackage{amssymb}
\usepackage{amscd}
\usepackage{caption}
\usepackage{indentfirst}
\usepackage{cleveref}

\begin{document}

\def\spacingset#1{\renewcommand{\baselinestretch}%
{#1}\small\normalsize} \spacingset{1}


\if1\blind
{
  \title{\bf Rethinking the Win Ratio: A Causal Framework for Hierarchical Outcome Analysis}
  \author{Mathieu Even\hspace{.2cm}\\
    Theremia and PreMeDICaL Inria-Inserm, University of Montpellier, France\\
    and \\
    Julie Josse \\
    PreMeDICaL Inria-Inserm, University of Montpellier, France}
  \maketitle
} \fi

\if0\blind
{
  \bigskip
  \bigskip
  \bigskip
  \begin{center}
    {\LARGE\bf Rethinking the Win Ratio: A Causal Framework for Hierarchical Outcome Analysis}
\end{center}
  \medskip
} \fi

\bigskip
\abstract{
 Quantifying causal effects in the presence of complex and multivariate outcomes remains a key challenge in treatment evaluation. For hierarchical multivariate outcomes, the FDA recommends the Win Ratio and Generalized Pairwise Comparisons approaches \citep{Pocock2011winratio,Buyse2010}. However, commonly used  estimators  can yield treatment recommendations that target a population-level estimand (the probability that a randomly sampled patient under treatment fares better than another randomly sampled patient under control), which can contradict conclusions drawn from an ideal estimand (the probability that an individual would fare better with treatment than without), especially in heterogeneous populations.
   This discrepancy arises from the non-identifiability of the latter estimand and underscores both the  influence of the chosen causal measure on the resulting conclusions and the necessity of articulating the underlying causal framework with clarity. 
    We propose a novel, individual-level yet identifiable causal effect measure that more closely approximates the ideal individual-level estimand. We show that computing the Win Ratio or Net Benefit via nearest-neighbor pairing between treated and control patients, which can be seen as an extreme form of stratification, yields an estimator of our new causal measure in both randomized controlled trials and observational settings. We then develop a distributional regression framework, alongside semiparametric efficient estimators. Our methods are simple to implement and readily applicable in practice.
    We evaluate the proposed approach through simulations and apply it to the CRASH-3 trial \citep{crash3}, a major study assessing the effects of tranexamic acid in patients with traumatic brain injury.
}

\maketitle

\renewcommand\thefootnote{}

\renewcommand\thefootnote{\fnsymbol{footnote}}
\setcounter{footnote}{1}

\section{Introduction}

Quantifying the benefit of a treatment in clinical research can be challenging, especially when outcomes are complex, multidimensional, or involve competing risks.
Traditional statistical approaches often struggle to capture the nuanced relationships between such outcomes, limiting their ability to provide clinically meaningful insights.
In these cases, innovative methods are required to address the inherent complexity of the data, to go beyond considering a single composite summary outcome.
The Win Ratio \citep{Pocock2011winratio} and Generalized Pairwise Comparisons \citep{Buyse2010} have emerged as powerful tools to evaluate treatment effects by comparing groups through hierarchical and multidimensional assessments of outcomes,
to the point where they appear in the recent Food and Drugs Administration (FDA) guidances for handling multiple 
outcomes (see the FDA report \emph{Multiple Endpoints in Clinical Trials Guidance for Industry, 2022}).

The Win Ratio and Net Benefit methodologies \citep{Pocock2011winratio,Buyse2010} form pairs of patients, each pair consisting of a patient in the control group and a treated patient.
Each pair is then considered as a ‘‘Win'' if the outcome of the treated patient is considered as more favorable than the control one, as a ‘‘Loss'' if it is considered as less favorable, and as a ‘‘Tie'' if the two patient outcomes are comparable.
We here illustrate the hierarchical comparison process, drawing inspiration from examples in cardiovascular trials \citep{Pocock2011winratio,Redfors2020}.
These studies evaluate multiple endpoints consisting of death, stroke, and heart failure hospitalizations (HFH), prioritizing events based on their clinical severity.
In this hierarchy, death is considered the most severe outcome, followed by stroke, and finally HFH.
For each patient pair, comparisons proceed as follows.
\emph{(i)}
Determine which patient died first during their shared follow-up period.
If one patient died earlier, the other patient is deemed the "winner" for this pair.
\emph{(ii)}
If neither patient died, assess who experienced a stroke first.
The patient with the later or no stroke is considered the "winner."
\emph{(iii)}
If neither patient died nor had a stroke, compare the number of HFHs during follow-up.
The patient with fewer HFHs is declared the "winner."
This hierarchical process stops at the first event that distinguishes between a win or loss for the pair.
If no event differentiates the pair, the outcome is recorded as a tie: depending on the methodology used, a tie can then be counted as a loss, as 1/2 instead of 1 or 0 (for respectively win or loss), or simply discarded.
This approach ensures that clinically meaningful priorities are respected while maximizing the utility of the available data.
For $Y_i$ and $Y_j$  the (multidimensional) outcomes of two patients $i$ and $j$, respectively treated and control, we write
\[Y_i\succ Y_j\,,\]
for a win of $i$ over $j$ and $Y_i\sim Y_j$ for a tie.
The Win-Proportion, the Win Ratio and the Net-Benefit of the treatment are then defined as 
\begin{equation*}
	\mathrm{Win\, Proportion}= \frac{\#\mathrm{Wins}}{\#\mathrm{Pairs}}\,,\quad \mathrm{Win\, Ratio}= \frac{\#\mathrm{Wins}}{\#\mathrm{Losses}}\,,\quad \mathrm{Net\, Benefit}= \frac{\#\mathrm{Wins}-\#\mathrm{Losses}}{\#\mathrm{Pairs}}\,.
\end{equation*}
Treatment recommendations are then made using these computed values.
If the win proportion is (significantly) above 0.5, treatment should be preferred over non-treatment, while for the Win Ratio and the Net Benefit the threshold values are respectively 1 and 0.


Considering different pairings between treated and control patients can lead to distinct treatment recommendations.  As illustrated in \Cref{fig:comp_WR_all_NN_adv}, even in a simple synthetic setting with randomized treatment assignment, complete pairings—where every treated patient is compared with every control—and nearest-neighbor matching on covariates yield opposite recommendations. This phenomenon arises from two key subtleties: first, different treatment-effect measures can  lead to different clinical decisions; second, the asymptotic behavior of pairwise comparisons is highly sensitive to the choice of pairing scheme. Clarifying which causal estimand is being targeted and understanding its properties are therefore essential for selecting an appropriate methodological approach.

In order to better understand the behavior of generalized pairwise comparison methods such as the Win Ratio, several works have put these approaches into a causal inference framework. This line of work draws inspiration from Mann-Whitney-Wilcoxon comparisons \citep{Wilcoxon1945,Mann1947}, that test if a univariate random variable is stochastically larger than another. 
An ideal treatment effect measure can be defined as the probability that a given individual fares better with than without treatment, by applying Mann-Whitney-Wilcoxon tests the counterfactual of a given individual.
Formally, given $Y_i(1)$ and $Y_i(0)$ the potential outcomes of a given patient with and without treatment \citep{splawa1990application}, this ideal individual-level estimand writes as:
\begin{equation}\label{eq:intro_indiv}
	\proba{Y_i(1)\succ Y_i(0)}\,.
\end{equation}
However, this quantity is \emph{non-identifiable}: it depends on the joint distribution of the potential outcomes, which is \textit{never} observed.
The estimand in \Cref{eq:intro_indiv} has then been mentioned in a series of works \citep{Mao2017,Guo2022,Chen2024,Yin2022,zhang2022causalinferencewinratio,Chiaruttini2024}, that circumvent its non-identifiability by instead estimating a \textit{population-level} measure, that can be defined as:
\begin{equation}\label{eq:intro_pop}
    \proba{Y_i(1)\succ Y_j(0)}\,,
\end{equation}
where $i$ and $j$ are two independent patients sampled randomly. Under standart causal assumptions, this causal measure is identifiable.

As we highlight in subsequent sections, the population-level measure $\proba{Y_i(1)\succ Y_j(0)}$ does not capture heterogeneity in the population, and even in randomized control trial settings, might favor treatment although a minority benefits from it.

In this paper, we introduce a third causal effect measure for quantifying treatment effects using generalized pairwise comparisons—one that is both individual-level and identifiable. Rather than substituting in \Cref{eq:intro_indiv} the unobserved counterfactual with the observed outcome of an independently sampled patient, as done in the population-level measure of \Cref{eq:intro_pop}, our approach replaces it with the counterfactual of another patient conditional on sharing the same covariates. This captures individual-level comparisons while preserving identifiability.
Informally, if patients $i$ and $j$ have features $X_i$ and $X_j$, the newly introduced causal estimand writes as:
\begin{equation}\label{eq:intro_cond}
    \proba{Y_i(1)\succ Y_j(0)|X_i=X_j}\,.
\end{equation}
Our paper is thus devoted to introducing formally this causal effect measure, to discuss its properties, compare it with the two other existing measures in this context introduced above, and provide estimation strategies.
In particular, and most importantly, we show that traditional generalized pairwise comparisons based on pairing strategies can be combined with covariate-matching methods in causal inference \citep{abadie2006large}, to estimate \Cref{eq:intro_cond}, both in RCT and observational settings. 

\subsection{Contributions and outline of the paper.}

\begin{figure}    
    \centering
    \begin{minipage}{0.6\textwidth}
        \centering
			\centering
  \includegraphics[width=\linewidth, trim=0 145 0 170, clip]{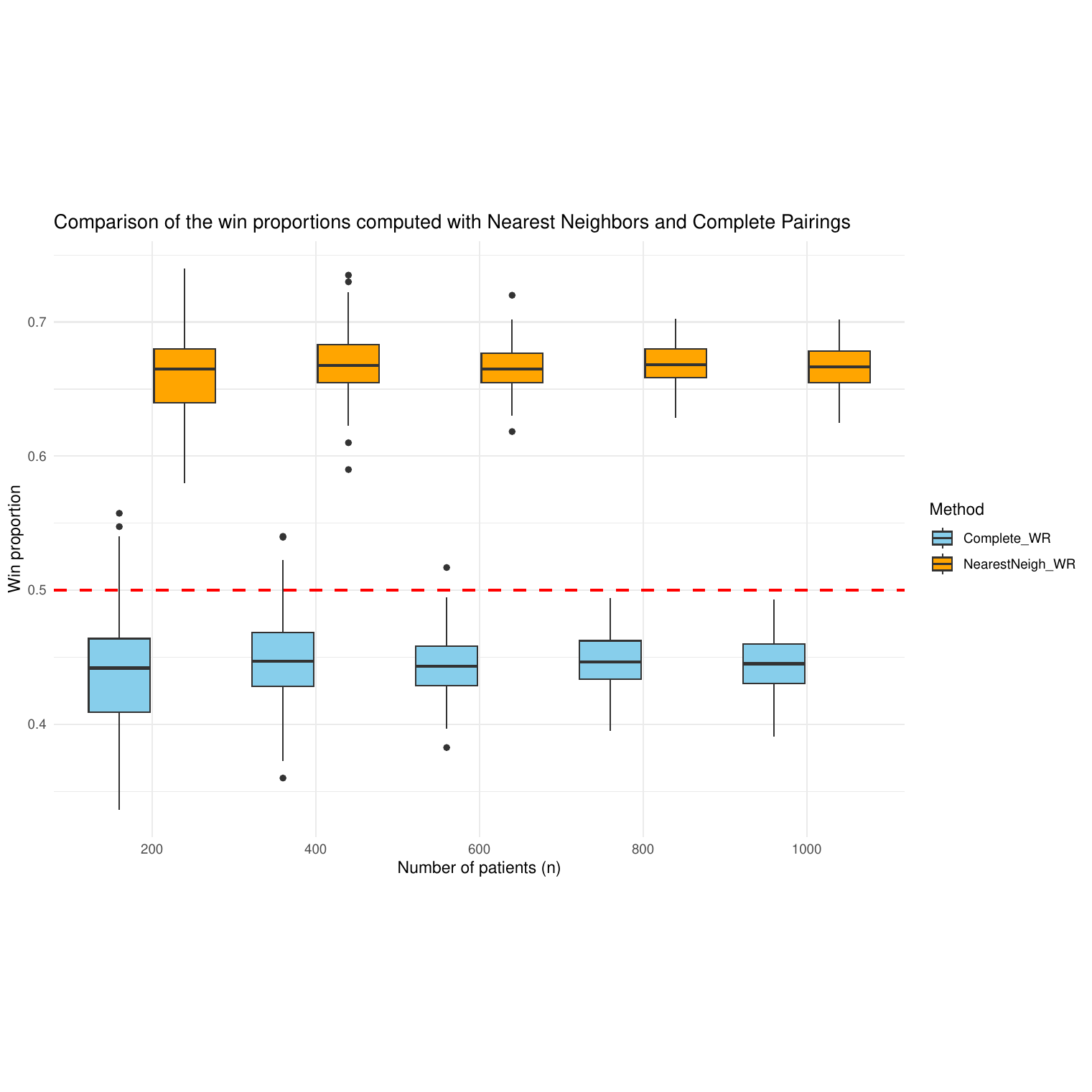}
		    \end{minipage}%
    \begin{minipage}{0.3\textwidth}
        \caption{Comparison of the win proportion $p_\mathrm{W}$ computed with complete pairings and Nearest Neighbor pairings.
		Setting of \Cref{ex:counter_emp}.
        Boxplots over 100 runs.
        The two approaches lead to different treatment recommendations (above and below 0.5).
		}
\label{fig:comp_WR_all_NN_adv}
    \end{minipage}
\end{figure}

We introduce our new causal estimand in \Cref{sec:estimands}, along with the causal assumptions and notations used throughout. Its main properties are detailed in \Cref{sec:misc_remarks}, where we compare it with the two existing measures (\Cref{eq:intro_indiv,eq:intro_pop,eq:intro_cond}). In short, our estimand is collapsible, interpretable as the mean of individual-level effects, and more robust to population heterogeneity and unknown strata—features illustrated through a synthetic example inspired by \Cref{fig:comp_WR_all_NN_adv}.

In \Cref{sec:matchings}, we study two natural pairing strategies in randomized trials—complete pairings and nearest-neighbor matching—and show they consistently estimate the population-level (\Cref{eq:intro_pop}, a result that was already known) and our proposed estimand (\Cref{eq:intro_cond}), respectively. We then extend the analysis of nearest neighbor pairing to the observational setting, and show that it naturally inherits the properties of covariate-matching strategies to balance for observed confounders \citep{abadie2006large,zubizarreta2012using}.
The matching estimator is thus consistent in both RCT and observational settings.

We then depart form matching methods and from the existing literature, by developing estimation methods that directly rely on estimating nuisance parameter functions (\Cref{sec:nuisance}), to face the fact that Nearest Neighbors may be slow to converge in the presence of high dimensional input features and their lack of robustness in the presence of missing data in the covariates. We propose a distributional regression approach using distributional random forests (\Cref{sec:distributional_regression}), and further correct first-order biases by  using the efficient influence function, yielding a semiparametric efficient estimator (\Cref{sec:EIF}).


Finally, we illustrate our methodology on synthetic observational data and the CRASH-3 trial \citep{crash3}, a large multi-center RCT that included over 12,000 patients in 175 hospitals across 29 countries. In \Cref{sec:crash3},  we compare existing approaches with our own and provides practical guidance for estimation and inference. Whereas traditional methods produce non-significant results (confidence intervals including 1), our individual-level estimand leads to statistically significant conclusions—highlighting the impact of the estimand choice.

\subsection{Related works}

\paragraph{Hierarchical Outcome Analysis and practical advancements.}
\citep{Pocock2011winratio,Buyse2010} introduced the Win Ratio and Generalized Pairwise Comparisons, extensions of \emph{Wilcoxon-Mann-Whitney tests} \citep{Wilcoxon1945,Mann1947}, which test whether $Y$ is stochastically larger than $Z$ via $\proba{Y \geq Z}$. Originally developed in the context of heart failure studies \citep{Pocock2011winratio}, the method continues to be actively investigated and refined in this area \citep{cunningham2026win,lopes2026sacubitril}.
These methods have been applied to hierarchical outcomes: e.g., \citep{Pocock2023} in the EMPULSE trial, showing empagliflozin’s benefit (Win Ratio 1.38); \citep{backer2024design} in acute promyelocytic leukemia, prioritizing efficacy over tolerability; and \citep{boentert2024applying} in the COMET trial on Pompe disease.
\citep{Redfors2020} review design and reporting considerations, while \citep{Ajufo2023} caution that “wins” may lack clinical meaning, ties are often ignored, and patient-reported outcomes or risk stratification may bias results.
Further limitations include censoring issues \citep{Mao2024defining_estimands} and unaddressed missing covariate data.

Finally, and closest to our work, \citep{dong2018stratified} proposed a \emph{stratified Win Ratio} approach, inspired by stratified odds ratio approaches \citep{cochran1954some,mantel1959statistical}, with follow-up works further developing the interpretation of stratification in hierachical outcome comparisons \citep{gasparyan2021adjusted,dong2023win,matsouaka2022robust}.
Our approach and results in this paper tend towards recommending the use of as much stratification as possible when dealing with Win Ratio and more generally comparison-based estimators with heterogeneous populations.

\paragraph{Formalization of hierarchical outcome analyses.}
Beyond empirical guidance, recent work formalizes estimands for hierarchical comparisons and extends them to observational data via Wilcoxon-Mann-Whitney tests \citep{Wilcoxon1945,Mann1947}. 
\citep{Mao2017} recast pairwise comparisons as $U$-statistics and proposed (augmented) inverse propensity weighting to estimate the probability that a treated patient “wins” against a control. 
This population-level estimand has since been studied for rank-sum tests \citep{Chen2024,Yin2022} and contrast functions \citep{Guo2022}, and further applied in cancer studies \citep{Chiaruttini2024} or extended to dependent or clustered data \citep{zhang2022causalinferencewinratio,Zhang2021}.

\paragraph{Relation to PNS and PIMs.}
Our work on hierarchical outcomes, with generalized win functions (\textit{i.e.,} replacing $\one_{y\succ y'}$ by some general comparison function $w(y,y')$) is also related to works done arount \textit{Probability of Necessity and Sufficiency} (PNS) and treatment effects with ordinal outcomes.
PNS is the probability that treatment improves the outcome, and is necessary for doing so.
Concretely, the PNS writes as $\mathrm{PNS} = \proba{Y_i(1)=1,Y_i(0)=0}$, if outcomes are \textit{binary} and $Y_i=0$ corresponds to a worse outcome that $Y_i=1$.
For general real-valued outcomes, the PNS will then write as $\mathrm{PNS} = \proba{Y_i(1)>Y_i(0)}$, which is actually a subcase of the individual-level yet unidentifiable causal measure previously defined in \Cref{eq:intro_indiv}.
The average is here taken over all individuals.
From the PNS, population-level decisions can be made, by quantifying the average individual benefits of treatment over the whole population.
Due to unidentifiability of the PNS, \citep{mueller2023personalized} instead resort to providing \textit{identifiable} bounds on the PNS (called probability of benefit), from combined RCT and observational data.
Thus, from such studies, a trialist will obtain $[a_\mathrm{PNS},b_\mathrm{PNS}]$, an interval in which the PNS lies in. That interval may be uninformative due to the non-identifiability of the PNS, but it may also lead to point intervals ($a_\mathrm{PNS}=b_\mathrm{PNS}$) in favorable cases.
The same approach also exists for \textit{ordinal outcomes} \citep{huang2017inequality,lu2018treatment}.
Our approach, that instead consists in defining an individual-level and identifiable measure, is orthogonal. Our outcomes and frameworks are more general than ordinal or real-valued outcomes, so that bounds like \citep{tian2000probabilities} are not possible in our setting ($\proba{Y_i(1)\succ Y_i(0)}$ cannot be upper or lower bounded by quantities such as $\esp{Y_i(1)}$, $\esp{Y_i(0)}$, etc, if outcomes are multivariate).

\citep{thas2012probabilistic} introduced Probabilistic Index Models (PIMs), that are defined as models for the Probabilistic Index statistics $\proba{Y\succ Y^*|X,X*}$, where $(Y,X)$ is a couple of observed outcome and features, and $(X^*,Y^*)$ is an independent copy of $(X,Y)$. If $X$ is the treatment, this becomes equivalent to the Mann-Whitney-Wilcoxon tests.
Formally, this PIM definition shares similarity with the way we define our causal estimand, in that an independent copy is introduced. PIMs are however conceptually different from what we introduce in terms of objectives, since we compare independent copies that are treated and control, but share same features.

\paragraph{Matching methods.}
Matching methods such as nearest neighbors are a core tool in causal inference for creating comparable treatment groups and mitigating confounding in observational studies \citep{stuart2010matching}. Advances include optimal matching with exact covariate balance constraints \citep{zubizarreta2012using}, theoretical results on large-sample properties and bias correction \citep{abadie2006large}, and recent overviews of modern extensions and applications \citep{wong2024handbook}. The approach developed in this paper generalizes covariate-based matching for estimating the average treatment effect with the risk difference as in \citep{abadie2006large}, to estimating the hierarchical outcomes-based causal measures we introduce.
Importantly, note that covariate-based matching in these papers aims at adjusting for confounders: nearest-neighbor thus only needs to be performed on these confounding variables.
In our context, as discussed further in \Cref{sec:misc_remarks}, matching also needs to be performed on treatment effect modifiers.

\section{Causal inference framework for Win Ratio and generalized pairwise comparisons}

\subsection{Definitions and assumptions}

We assume access to $n$ i.i.d.\ samples, associated with feature vectors $X_i\in\cX$, treatment $T_i\in\{0,1\}$, and observed outcome $Y_i\in\cY$ (possibly multivariate, e.g.\ $\cY\subset\R^d$). Under the \emph{potential outcome} framework \citep{splawa1990application}, each patient has two outcomes $Y_i(0),Y_i(1)\in\cY$, only one of which is observed.
We make the \textit{SUTVA}, \emph{unconfoundedness} and \emph{positivity} assumptions.

\begin{assumption}[SUTVA]\label{hyp:sutva}
$Y_i = Y_i(T_i)$ for all $i\in[n]$.
\end{assumption}
\begin{assumption}\label{hyp:unconfoundedness}
$\{Y_i(0),Y_i(1)\} \perp T_i \mid X_i$.
\end{assumption}

\begin{assumption}\label{hyp:positivity}
There exists $\eta\in(0,1)$ such that $\eta \le \pi(x) \le 1-\eta$ for all $x\in\cX$,
where $\pi(x) = \proba{T_i=1 \mid X_i=x}$.
\end{assumption}

We want to know if a given patient would fare better under treatment than without it.
In the potential outcomes framework, there exists several causal measures to quantify this, amongst which the \emph{Risk Difference} (RD) if $\cY\subset\R$, for which the Average Treatment Effect (ATE) writes as $\tau_\mathrm{RD}= \esp{Y_i(1)-Y_i(0)}$.
However, we would like to handle more general outcomes and thus define new and more general causal estimands in \Cref{sec:estimands}.
Given two outcomes $y,y'\in\cY$, treatment effects cannot always be quantified directly with the sign and magnitude of their difference $y-y'$. Instead, we consider a more general ``win function'', that quantifies preferences.
Our framework generalizes the \emph{lexicographic order}\footnote{The \emph{lexicographic order} is a total order on $\R^d$, defined as $y\succ y'$ if and only if $\inf\set{k\in[d],y_k>y_k'}<\inf\set{k\in[d],y_k<y_k'}$, with the convention that the infimum of an empty set is $+\infty$.
In plain words, the \emph{lexicographic order} amounts to order vectors as words in the dictionary.}
beyond the setting introduced by \citep{Pocock2011winratio} for hierarchical outcomes.

\begin{definition}[Win function]
	Let 
	\begin{equation*}
		w\,:\,(y,y')\in\cY^2 \mapsto w(y,y')\in[0,1]\,,
	\end{equation*}
	be the \emph{win function}, taking two outputs $y,y'$ to compare, and outputing a value between 0 and 1.
\end{definition}

A typical example is:
\begin{equation}\label{eq:ex_win_proba}
	w(y|y')=\left\{     
	\begin{aligned}
		& 1\quad &\text{if}& \qquad y\succ y'\\
		&\frac{1}{2}\quad &\text{if}& \qquad y\sim y'\\
		&0 \quad &\text{if}&\qquad y\prec y'
	\end{aligned}
	\right.\qquad,
\end{equation}
where $\succ$ is an order on $\cY$, that we refer to as \emph{clinical order}, and $\sim$ means that the outcomes are similar or cannot be compared.
If ‘‘ties'' are discarded, the win function writes as:
\begin{equation}\label{eq:ex_win_proba_wo_ties}
	w(y|y')=\left\{     
	\begin{aligned}
		& 1\quad &\text{if}& \qquad y\succ y'\\
		&0 \quad &\text{if}&\qquad y\preceq y'
	\end{aligned}
	\right.\qquad,
\end{equation}
\citep{Pocock2011winratio}
first introduced the Win Ratio to handle composite endpoints in clinical trials based on clinical priorities, without formalizing it using the potential outcomes framework.
In their example, outcomes $Y_i$ are of dimension $2$ and potential outcomes
$Y_1(t),Y_2(t)\in\R\cup\set{\infty}$ would respectively correspond to the time after treatment (or non treatment) before an eventual cardiovascular death event, and the time after (non-)treatment before an eventual heart-failure hospitalization.
$Y_i(t)=\infty$ means that no such event occurred.
The win function here writes as: $w(y|y')=\one_\set{y_1>y_1'}+\one_\set{y_1=y_1'\,,\,y_2>y_2'}$, \emph{i.e.} $w(y|y')=1$ if and only if $y$ is strictly larger than $y'$ for the lexicographic order.
Note that in this example, ties are counted as losses, which may not always be the case,
as this may cause problems when ties should not be discarded nor treated as losses, as pointed out by \citep{Ajufo2023}.

\paragraph{Notations.}
For a sequence $(Z_n)_{n\geq 0}$ of random variables with values in a metric space $(\cZ,d)$ and $Z$ a random variable in $\cZ$, we say that $Z_n$ converges in probability towards $Z$ if for any $\eps>0$ we have that $\proba{d(Z_n,Z)>\eps}\to0$.
We say that a measurable event $\cE$ is almost sure if its probability is 1.
We say that $Z_n$ converges almost surely towards some value $\ell$ if the event $\set{Z_n\to \ell}$ is almost sure.
$Z_n$ is a consistent estimator of some quantity $\ell$ if $Z_n$ converges in probability towards the constant random variable $\ell$.

\subsection{Traditional Win Ratio, Net-Benefit and Mann-Whitney-Wilcoxon comparison estimators}

Let $\cN_t=\set{i:T_i=t}$ for $t=0,1$ be respectively the control and treated groups.
Below, we formalize the estimators used by \citep{Pocock2011winratio,Buyse2010} for Win Ratio and Generalized Comparisons.
We refer to these estimators as \emph{traditional} or \textit{historical} Win Ratio and Net-Benefit.
\citep{Pocock2011winratio,Buyse2010} form pairs $\cC\subset \cN_1\times \cN_0$ and define 
\begin{equation}\label{eq:nb_wins}
	n_\mathrm{W} = \sum_{(y,y')\in\cC}w(y|y')\,,
\end{equation}
as the number of wins and $n_\mathrm{L}=|\cC|-n_\mathrm{W}$ as the number of losses.
The \emph{Win Proportion}, \emph{Win Ratio} \citep{Pocock2011winratio} and \emph{Net Benefit} \citep{Buyse2010} are defined as:
\begin{equation}\label{eq:win_prop}
	\hat p_\mathrm{W}=\frac{n_\mathrm{W}}{|\cC|}\,,\qquad \hat R_\WR = \frac{n_\mathrm{W}}{n_\mathrm{L}}\,,\qquad \hat\Delta_\mathrm{NB}= 2\hat p_\mathrm{W} -1=\frac{n_\mathrm{W}-n_\mathrm{L}}{|\cC|}\,.
\end{equation}
Defining a loss proportion as $\hat p_\mathrm{L}=\frac{n_\mathrm{L}}{|\cC|}$, note that we have that the Win Ratio writes as $\hat p_\mathrm{W}/\hat p_\mathrm{L}$ while the net benefit writes as $\hat p_\mathrm{W}-\hat p_\mathrm{L}$.

\paragraph{Pairing options.}
The choice of the set of pairs $\cC\subset\cN_0\times \cN_1$ of control and test individuals used to compute the number of wins in \Cref{eq:nb_wins} 
has a crucial impact on the quantity being computed.
The pair set $\cC$ might vary from the two natural following extremes:
\begin{enumerate}
	\item \emph{Complete pairings}, for which we have:
	\begin{equation}\label{eq:pair_complete}
		\cC_\mathrm{Tot}=\cN_0\times \cN_1\,.
	\end{equation}
	Complete pairings are the prevalent strategy in Win Ratio or Generalized Pairwise Comparisons analyses.
	\item \emph{Nearest-neighor pairings}, for which
	\begin{equation}\label{eq:pair_nn}
		    \begin{aligned}
&\cC_\mathrm{NN}=\cC_\mathrm{NN}^{(1)}\cup\cC_\mathrm{NN}^{(0)} \,,\\
\text{where} \quad 
&\cC_\mathrm{NN}^{(1)}=\set{(i,\sigma_t^\star(i))\,|\,i\in\cN_1}\\
\text{and}\quad &\cC_\mathrm{NN}^{(0)}=\set{(\sigma_t^\star(j),j)\,|\,j\in\cN_0}\,,
        \end{aligned}
	\end{equation}
	where $\sigma_t^\star(i):\cN_t\to\cN_{1-t}$ for $t\in\set{0,1}$ matches $i\in\cN_t$ to the treated or control individual $j\in\cN_{1-t}$ that has closest features in $\cX$ \emph{i.e.},
	$\sigma_t^\star(i)\in\argmin_{k\in\cN_{1-t}}\NRM{X_i-X_k}^2$.
	$\cC_\mathrm{NN}$ can also be generalized to $k-$Nearest Neighbors.	
\end{enumerate}

As highlighted in \Cref{ex:counter_emp}, using different pairings leads to different Win Proportions (and thus to different Win Ratios and Net Benefits), to the point where treatment recommendations may even differ.

\begin{example}\label{ex:counter_emp}
	Suppose that we have $n=6$ individuals with univariate and real outcomes ($\cY\subset \R$).
	Assume that for $i=1,2,3,4$, individuals are men (for which $X_i=0$) and we have $Y_i(1)=y_1>y_0=Y_i(0)$ while for $i=5,6$ individuals are women (for which $X_i=1$) and we have $Y_i(1)=y_1<y_0=Y_i(0)$, and that for $i\in\set{1,\ldots,6}$ we have $T_i=1$ if $i$ is an odd number.
	Assume then that $y_0'>y_1>y_0>y_1'$.
	Then, we have:
	\begin{equation*}
			\hat p_\mathrm{W}= 
\left\{		\begin{aligned}
	&\frac{2}{3}&\quad &\text{if}&\quad &\cC=\cC_\mathrm{NN}\\
	&\frac{4}{9}&\quad &\text{if}&\quad &\cC=\cC_\mathrm{Tot}\\
\end{aligned}\right.\quad,
	\end{equation*}
	leading to $\hat p_\mathrm{W}>1/2$ or $\hat p_\mathrm{W}<1/2$ and thus to different treatment decisions depending on the coupling pairs chosen.
	Here, treatment favors $4$ out of the $6$ patients, while no-treatment only favors $2$ out of the $6$: complete pairings thus favors the treatment option that benefits to only a minority of patients.
	This is further illustrated in \Cref{fig:comp_WR_all_NN_adv}
	in the setting of \Cref{ex:counter_emp},
	with $n/3$ women and $2n/3$ men, in a
    RCT setting with treatment probability of $1/2$.
\end{example}

\subsection{From estimators to causal measures and estimands}\label{sec:estimands}

The quantities introduced so far --- $\hat p_\mathrm{W},\hat R_\WR,\hat \Delta_\mathrm{NB}$ ---, are data-dependent estimators (hence the $\hat\cdot$ notation).
To efficiently capture treatment effects and treatment comparisons, we need to first answer the following crucial question:
\emph{what is the estimand that these estimators seek at estimating?}

\paragraph{A natural but non-identifiable causal estimand.}
To determine if individuals would fare better if treated or non-treated, we consider  a contrast function $w$, where for $y,y'\in\cY$, $w(y|y')$ quantifies the relative favorability of $y$ compared to $y'$.
This leads to consider the following causal effect measure, that compares the two potential outcomes of a given individual using the contrats/win function $w$:
\begin{equation}\label{eq:tau_indiv}
	\tau_\mathrm{indiv}=\esp{w(Y_i(1)|Y_i(0))}\,.
\end{equation}
If $w$ is respectively as in \Cref{eq:ex_win_proba} and \Cref{eq:ex_win_proba_wo_ties}, we have:
\begin{equation*}
	\tau_\mathrm{indiv}=\proba{Y_i(1)\succ Y_i(0)}+\frac{1}{2}\proba{Y_i(1)\sim Y_i(0)} \,,\qquad\text{and}\qquad \tau_\mathrm{indiv}=\proba{Y_i(1)\succ Y_i(0)}\,.
\end{equation*}
However, as highlighted by several works \citep{Mao2017,Guo2022,Chen2024,Yin2022,zhang2022causalinferencewinratio,Chiaruttini2024}, this causal measure is \emph{non-identifiable}
since estimating it in general requires the knowledge of the \emph{joint distribution} of the potential outcomes, which is never observed
\footnote{
	If $Y_i(0),Y_i(1)\sim\mathrm{Bernoulli(1/2)}$, the ATE with the risk difference $\esp{Y_i(1)-Y_i(0)}$ is identifiable (via e.g. taking the mean on test and control groups, in a RCT setting),
	while the ATE with $w(y|y')=\one_\set{Y_i(1)>Y_i(0)}$, that writes as $\proba{Y_i(1)>Y_i(0)}$, is not identifiable.
	Indeed, in that latter case, coupling $(Y_i(1),Y_i(0))$ as $Y_i(1)=Y_i(0)$ gives $\proba{Y_i(1)>Y_i(0)}=0$, while taking independent potential outcomes leads to $\proba{Y_i(1)>Y_i(0)}=\frac{1}{4}$.
	Since the distribution of the observations $(X_i,T_i,Y_i)$ does not change by taking either coupling but the value of $\esp{w(Y_i(1)|Y_i(0))}$ changes, we have non-identifiability: $\esp{w(Y_i(1)|Y_i(0))}$ cannot be expressed as a function of the law of the observations.
}.
$\tau_\mathrm{indiv}$ is indeed an \emph{individual-level} causal estimand \citep{Fay2024}, as it is directly a function of the joint probability distribution $\cP(\set{Y_i(0),Y_i(1)})$,
as opposed to \emph{population-level} causal estimands that are functions of $\set{\cP(Y_i(0)),\cP(Y_i(1))}$.
With the Risk Difference, we would have $w(y|y')=y-y'$, and the $\tau_\mathrm{indiv}$ would be equal to the ATE $\tau_\mathrm{ATE}=\esp{w(Y_i(1)|Y_i(0))}=\esp{Y_i(1)-Y_i(0)}$, which is identifiable thanks to the linearity of the contrast function.

\paragraph{A population-level causal estimand.}
To circumvent this non-identifiability issue of individual-level causal measures, \citep{Mao2017,Guo2022,Chen2024,Yin2022,zhang2022causalinferencewinratio,Chiaruttini2024} consider a population-level causal measure
$\tau_\mathrm{pop}$ instead of the indivual-level one.
Their causal measure writes as:
\begin{equation}\label{eq:win_pop}
	\tau_\mathrm{pop}= \esp{w(Y_i(1)|Y_j(0))}\,,
\end{equation}
where $i$ and $j$ are two different and independent individuals.
Considering two independent individuals leads to a different measure, that can now be estimated (see \Cref{sec:matchings}).

\paragraph{Our individual-level and identifiable causal estimand.}
Using $\tau_\mathrm{pop}$ leads to comparing individuals that may not be comparable, hence the following causal measure we introduce.
It is an identifiable relaxation of $\tau_\mathrm{indiv}$, defined by comparing patient $i$ with features $X_i$ to an independent copy \emph{that has the same features}.
Our measure $\tau_\star$ is an individual-level causal estimand.
We formalize this by using an independent copy of the counterfactual outcomes conditioned on the covariates: this does not require any stronger assumption, since an independent copy of a given random variable always exists.
Well-posedness, identifiability under classical assumptions and other misceleanous remarks are presented just after in \Cref{sec:misc_remarks}.

\begin{definition} \label{def:win_indiv} 
	For $x\in\cX$, let $\set{Y^{(x)}(0),Y^{(x)}(1)}$ be an independent copy of $\set{Y_i(0),Y_i(1)}|X_i=x$.
	Let
	\begin{equation*}
		\tau_\star(x)=\esp{w(Y^{(x)}(1)|Y_i(0))|X_i=x}\,,
	\end{equation*}
	and define
	\begin{equation}\label{eq:win_proba}
		\tau_\star = \esp{w(Y^{(X_i)}(1)|Y_i(0))}\,. 
	\end{equation}
\end{definition}

The quantities defined $\tau_\star(x)$ and $\tau_\star$ are respectively a conditional effect measure and a causal effect measure \citep{pearl2009causality} (equivalent of CATEs and ATE, respectively).
We now define the \emph{statistical estimands} related to our causal measures and to our causal estimands.
Recall that statistical estimands are functions of measurable quantities; as such, they cannot make appear counterfactual quantities such as potential outcomes.
For instance, the statistical estimand related to the causal estimand $\tau_\mathrm{RD}$ is $\esp{Y_i|T_i=1}-\esp{Y_i|T_i=1}$.
For the population-level causal measure $\tau_\mathrm{pop}$, the related statistical estimand writes as:
\begin{equation}\label{eq:estimand_pop}
	\esp{\esp{w(Y_i|Y_j)|T_i=1,T_j=0,X_i,X_j}}\,,
\end{equation}
while the statistical estimand related to $\tau_\star$ writes as
\begin{equation}\label{eq:estimand_star_1}
	\esp{\esp{w(Y^{(X_i)}(1),Y_i)|T_i=0,X_i}}\,,
\end{equation}
or, equivalently:
\begin{equation}\label{eq:estimand_star_2}
	\esp{\esp{w(Y_i,Y^{(X_i)}(0))|T_i=1,X_i}}\,.
\end{equation}
Under \Cref{hyp:sutva,hyp:positivity,hyp:unconfoundedness}, we have that these statistical estimands are equal to their associated causal estimands.
The relation with pairing approaches thus appear clear, after writing the statistical estimands: 
in a RCT, pairing all control and treated pairs and averaging the results will naturally estimate \Cref{eq:estimand_pop}.
On the other hand, pairing for each control unit its closest treated unit will estimate \Cref{eq:estimand_star_1}, while pairing for each treated unit its closest control unit will estimate \Cref{eq:estimand_star_2}. These two quantities being equal, the asymptotic behavior of the pairing strategies introduced in \Cref{eq:pair_nn} are equivalent.
This discussion is pursued more formally in \Cref{thm:consistency_WR} with consistency results.
In observational settings, the nearest neighbor pairing ($\cC_\mathrm{NN}$, $\cC_\mathrm{NN}^{(0)}$ or $\cC_\mathrm{NN}^{(1)}$) will have an impact on the asymptotic behavior of the resulting estimator, that will estimate either the ATE, the ATC or the ATT (average treatment effect on the control or treated).

\subsubsection{Comparisons of the three different causal effect measures}

We now discuss the differences of the three different causal effect measures introduced, to better grasp the novelty of ours ($\tau_\star$, \Cref{def:win_indiv}).

\paragraph{Collapsibility.}
A first motivation for preferring our measure $\tau_\star$ over the population-level counterpart $\tau_\mathrm{pop}$ is collapsibility. Indeed, $\tau_\star=\esp{\tau_\star(X_i)}$ by definition, making it directly collapsible, and highlighting $x\mapsto\tau_\star(x)$ as a natural conditional treatment effect.
On the other hand, although $\tau_\mathrm{pop}$ can also be expressed as the average of conditional effects—specifically, $\tau_\mathrm{pop} = \mathbb{E}[\tau_\mathrm{pop}(X_i)]$ where $\tau_\mathrm{pop}(x) = \mathbb{E}[w(Y_i(1) | Y_j(0)) \mid X_i = x]$—the conditional effect $\tau_\mathrm{pop}(x)$ still depends on the composition of the entire population, since in $\tau_\mathrm{pop}(x)$ an individual with covariates $x$ is compared to the whole population.
However, conditional effects should ideally reflect properties intrinsic to the individual given features rather than being influenced by the overall population. As a result, policies based on $\tau_\mathrm{pop}(X_i)$ are not truly personalized, since their recommendations are affected by the presence and characteristics of other individuals in the population.
This contradicts the rationale behind conditional effects and personalized recommendations, that should only be a function of the covariates, and not the whole population.
Note that $\tau_\star$ is also \emph{logic-respecting} \citep[Definition 6]{colnet2024riskratiooddsratio}.

\paragraph{Interpretation of the causal measure.}
Reporting results using different causal measures may lead to different conclusions or interpretations \citep{colnet2024riskratiooddsratio,Fay2024,Groenwold2011,Didelez2021}.
\citep{Fay2024} provides extended discussion on individual and population-level estimands and causal measures, such as the win proportion in our case or the risk-ratio that are not directly collapsible \citep{Fay2024,Groenwold2011,Didelez2021}.
However, for the risk ratio, the population-level estimand ($\esp{Y_i(1)}/\esp{Y_i(0)}$, instead of $\esp{Y_i(1)/Y_i(0)}$ for the individual-level risk ratio) can be interpreted in terms of treatment recommendations: $\esp{Y_i(1)}/\esp{Y_i(0)}=2$ means that treating everyone leads to an averaged outcome twice larger compared to treating no one.
In comparison-based estimands this is not the case: $\proba{Y_i(1)\succ Y_j(0)}=0.6$ does not mean that a majority of the population would benefit from being treated, as this depends on eventual population heterogeneity, for which we need to adjust for, which is precisely what our new estimand does. Next example illustrates this in the presence of unknown population stratas, on a synthetic example.

\paragraph{Hand's paradox for the Win Ratio \citep{Hand1992}.}
\Cref{remark:pop_vs_indiv} below is the estimand-wise version of \Cref{ex:counter_emp}.
It shows that the causal measures $\tau_\star$ and $\tau_\mathrm{pop}$ are not equivalent and may lead to different treatment recommendations if populations are heterogeneous.
In particular, \Cref{remark:pop_vs_indiv} justifies our preference and recommendations towards using $\tau_\star$ over $\tau_\mathrm{pop}$ if the population is heterogeneous.

\begin{example}\label{remark:pop_vs_indiv}
	Assume that outcomes are univariate ($Y_i(0),Y_i(1)\in\R$), and that $\cX=\cX_1\cup\cX_2$ with $\proba{X_i\in\cX_1}=1-\alpha$, $\proba{X_i\in\cX_2}=\alpha$ such that:
	\begin{equation*}
		Y_i(1)=y_1>y_0=Y_i(0)|X_i\in\cX_1\,,\quad Y_i(1)=y'_1<y'_0=Y_i(0)|X_i\in\cX_2\,,
	\end{equation*}
	almost surely.
	We then have, if $y_0'>y_1>y_0>y_1'$:
	\begin{equation*}
		\esp{w(Y^{(X_i)}(1)|Y_i(0))}=1-\alpha\,,\qquad \esp{w(Y_i(1)|Y_j(0))}=(1-\alpha)^2\,.
	\end{equation*}
	Thus, if $\frac{1}{2}>\alpha>1-\frac{1}{\sqrt{2}}$, we have:
	\begin{equation*}
		\esp{w(Y^{(X_i)}(1)|Y_i(0))}>\frac{1}{2}>\esp{w(Y_i(1)|Y_j(0))}\,,
	\end{equation*}
	leading to different conclusions in terms of treatment efficacy (see \Cref{fig:comp_WR_all_NN_adv}).
    This highlights the fact that Win Ratio or PNS measures favor treatments that benefit to the largest proportion of the population. Instead, the ATE with the RD takes into account the magnitude of the effect.
\end{example}

We highlight the fact that the phenomenon appearing in \Cref{remark:pop_vs_indiv} is \emph{not} reminiscent of Simpson's paradox \citep{simpson1951interpretation,wagner1982simpson}, as understood in the popular sense.
Simpson's paradox states a trend might appear in \emph{all} subgroups of a population, but still reverse when considering the average over all population.
Here, the paradox in \Cref{remark:pop_vs_indiv} is of a very different nature, since in both cases the average over the whole population is considered.
The difference lies in the way the average is taken: different causal measures might lead to different treatment effects.
Taking the average of individual effects over the global population (as done when considering $\tau_\star$) or comparing the whole treated group distribution with the control distribution (as done when considering $\tau_\mathrm{pop}$) can lead to opposite trends, as was previously noticed by \citep{Hand1992}.

\paragraph{Potential independence and interpretation as proxy measures.}
Under additional assumptions such as \emph{potential independence} (\emph{i.e.}, $Y_i(1)\indep Y_i(0)\,|\,X_i$), we have that $\tau_\star=\tau_\mathrm{indiv}$.
This assumption of potential independence is however quite strong and may be considered unlikely to hold true. 
It means that all that accounts for treatment effects is included in $X_i$, precluding \emph{e.g.} any unmeasured factor.		
Some examples (such as for instance comparing 2 different doses of a same treatment) make it impossible to assume conditional independence in general, hence the appeal of $\tau_\star$, since \Cref{def:win_indiv} does not need to make any assumption for $\tau_\star$ to be well-defined.

If we think of $\tau_\mathrm{pop}$ and $\tau_\star$ as a computable proxy\footnote{which is not necessarily why quantities like $\tau_\mathrm{pop}$ were originally introduced in the literature.} to approximate the ideal value $\tau_\mathrm{indiv}$ that cannot be approximated in general, we argue that, thanks to simple examples such as \Cref{remark:pop_vs_indiv}, $\tau_\star$ is a \emph{better} proxy than $\tau_\mathrm{pop}$, since it captures more information.
This is intuitively the case: both $(Y^{(X_i)}(1),Y_i(0))$ and $(Y_j(1),Y_i(0))$ are couplings of the random variables $Y_i(1)$ and $Y_i(0)$.
The coupling $(Y^{(X_i)}(1),Y_i(0))$ is however naturally closer to the coupling $(Y_i(1),Y_i(0))$ than the coupling $(Y_j(1),Y_i(0))$, since $(Y^{(X_i)}(1),Y_i(0))$ takes into account covariate effects, leading to:
\begin{equation*}
	d_{\ell^2}\big((Y^{(X_i)}(1),Y_i(0))\,,\,(Y_i(1),Y_i(0))\big)\leq d_{\ell^2}\big((Y_j(1),Y_i(0))\,,\,(Y_i(1),Y_i(0))\big)\,,
\end{equation*}
if the marginals are absolutely continuous with respect to the Lebesgue measure,
and where $d_{\ell^2}$ is the $\ell^2$ distance between densities.
Finally, next proposition formalizes the excess risk when using $\tau_\star$ or $\tau_\mathrm{pop}$ as proxis for $\tau_\mathrm{indiv}$.

\begin{proposition}
	We have:
	\begin{equation*}
		\left| \tau_\star - \tau_\mathrm{indiv}\right|\leq d_\mathrm{TV}(P_{Y^{(X_i)}(1),Y_i(0)},P_{Y_i(1),Y_i(0)})\,,
	\end{equation*}
	and
	\begin{equation*}
		\left| \tau_\mathrm{pop} - \tau_\mathrm{indiv}\right|\leq d_\mathrm{TV}(P_{Y_j(1),Y_i(0)},P_{Y_i(1),Y_i(0)})\,,
	\end{equation*}
	where $P_{Y^{(X_i)}(1),Y_i(0)},P_{Y_i(1),Y_i(0)},P_{Y_j(1),Y_i(0)}$ are respectively the joint distributions of $(Y^{(X_i)}(1),Y_i(0))$, $(Y_i(1),Y_i(0))$ and $(Y_i(1),Y_j(0))$, and $d_\mathrm{TV}$ is the total-variation distance between distributions.
	Furthermore, if the win function $w$ is $1-$Lipschitz, we have that:
	\begin{equation*}
		\left| \tau_\star - \tau_\mathrm{indiv}\right|\leq \cW_1(P_{Y^{(X_i)}(1),Y_i(0)},P_{Y_i(1),Y_i(0)})\,,
	\end{equation*}
	and
	\begin{equation*}
		\left| \tau_\mathrm{pop} - \tau_\mathrm{indiv}\right|\leq \cW_1(P_{Y_j(1),Y_i(0)},P_{Y_i(1),Y_i(0)})\,,
	\end{equation*}
	where $\cW_1$ is the $1-$Wasserstein distance between distributions.
\end{proposition}

\subsubsection{Dependence on the covariates, identifiability, and well-posedness}
\label{sec:misc_remarks}

\paragraph{Dependence of $\tau_\star$ over covariates:\textit{ which features should be included in $X$?}}
In order for $\tau_\star(x)$ and $\tau_\star$ to be identifiable, \Cref{hyp:unconfoundedness} must hold, requiring that all confounders must be included as covariates. We will here however require more than confounders.
\textit{Treatment effect modifiers} are features that have an influence on the treatment effect and that can change the value of $\tau_\star(x)$ if included. For instance, variables that can capture a hidden strata are treatment effect modifiers, as illustrated in \Cref{ex:counter_emp}.
In order for our matching between a patient $i$ and its independent copy conditioned on $X_i$ to be as meaningful as possible and for $\tau_\star$ to be as close to $\tau_\mathrm{indiv}$ as possible, we furthermore require treatment effect modifiers to also be included in the covariates, on top of confounders.
Note however that this is not restrctive and does not add more complexity than in classical treatment effect estimation: for the ATE with the RD, treatment effect modifiers can be included in the covariates in order to diminish variance of the estimates.

\paragraph{Identifiability under \Cref{hyp:sutva,hyp:positivity,hyp:unconfoundedness} and well-posedness of \Cref{def:win_indiv}.}
\Cref{eq:win_proba} can be identified, if it can be written as a function of the law $\cP_\mathrm{obs}$ of the observations $(X_i,T_i,Y_i)$.
Under SUTVA, this is the law of $(X_i,T_i,Y_i(T_i))$, while under unconfoundedness, $\dd\P(Y_i(t)=y|X_i)=\dd\P_\mathrm{obs}(Y_i=y|X_i,T_i=t)$.
Thus, since $\tau_\star(x)$ writes as:
\begin{equation*}
    \tau_\star(x) = \int w(y|y') \dd\P(Y_i(1)=y|X_i=x)\dd\P(Y_i(0)=y'|X_i=x)\,,
\end{equation*}
it can be expressed as a function of the distribution of the observations:
\begin{equation}\label{eq:id_tau_star}
    \tau_\star(x) = \int w(y|y') \dd\P_\mathrm{obs}(Y_i=y|X_i=x,T_i=1)\dd\P_\mathrm{obs}(Y_i=y'|X_i=x,T_i=0)\,.
\end{equation}
We thus have identifiability of $\tau_\star(x)$ and $\tau_\star$ through $\tau_\star=\esp{\tau_\star(X_i)}$.
Furthermore, a potential more intuitive way to understand $\tau_\star$ and $\tau_\star(x)$, is to write them, with a slight abuse of notations, as, where $i,j$ are two independent patients:
\begin{equation*}
    \tau_\star(x) = \esp{w(Y_i|Y_j)|X_i=X_j=x,T_i=1,T_j=0}\,,\quad \text{and}\quad \tau_\star = \esp{\esp{w(Y_i|Y_j)|X_i=X_j,X_i,T_i=1,T_j=0}}\,.
\end{equation*}
Note that $X_i=X_j$ is an event of mass equal to 0, so that conditioning on it requires some work as this quantity is not always well defined, hence our \Cref{def:win_indiv} that does not use this conditioning. 
However, this expression with the $X_i=X_j$ conditioning emphasizes the fact that our estimand indeed captures and extreme form of stratification.

Let us now briefly expose why the definition of $\tau_\star$ is always well-posed under no assumption.
Let $(E,\cE)$ and $(F,\cF)$ be two probability spaces.
$\nu:E\times \cF\to [0,1]$ is a transition kernel if it satisfies $\forall x\in E\,,\,\nu(x,\cdot)$ is a probability measure on $(F,\cF)$ and $\forall B\in\cF\,,\, \nu(\cdot,B)$ is $\cE-$measurable.
If $X\in \R^p$ and $Y\in\R^d$, the conditinal law of $Y$ given $X$ is a kernel $\nu$ on $(\R^d\times \cB(\R^p))$ that satisfies:
\begin{equation*}
    \P_{(X,Y)}=\P_{X\cdot \nu}\,,\quad \text{i.e.}\quad \forall (A,B)\in\cB(\R^p)\times \cB(\R^d)\,,\quad \proba{X\in A,Y\in B}=\int_{x\in A}\nu(x,B)\P_X(\dd x)\,.
\end{equation*}
Such a transition kernel always exists: this is Miloslav Jiřina's Theorem \citep{jirina} for Borelian measures and random variables.
We write 
\begin{equation*}
    \proba{Y\in B|X=x}= \nu(x,B)\,.
\end{equation*}
To build two independent copies of $Y$ given $X=x$, we thus draw $Y_x,Y'_x$ with:
\begin{equation*}
    \forall (B,B')\in\cB(\R^d)^2\,,\quad\proba{Y_x\in B\,,\, Y_x'\in B'}=\nu(x,B)\nu(x,B')\,.
\end{equation*} 
We thus have created a copy of $Y$ that satisfies the two following properties:
\emph{(i)} it is independent of $Y$ conditionally on $X$;
\emph{(ii)} conditionally on $X$, the distributions of $Y'$ and $Y$ are the same.
Their joint distribution with $X$ writes as:
\begin{equation*}
    \proba{X\in A\,,\, Y\in B\,,\,Y^{(X)}\in B'}=\proba{X\in A}\esp{\nu(X,B)\nu(X,B')|X\in A}\,,
\end{equation*}
for all $(A,B,B')\in\cB(\R^p)\times \cB(\R^d)^2$.

\section{Matching Strategies for Win Ratio estimation}
\label{sec:matchings}

\subsection{Consistency of traditional Win Ratio, Net-Benefit and Win Proportions in the RCT setting}
\label{sec:RCT}

In this section, we study $\hat p_\mathrm{W}$ in light of the causal measures we previously defined.
As expected from \Cref{ex:counter_emp,remark:pop_vs_indiv}, the behavior of $\hat p_\mathrm{W}$ crucially depends on the pairings considered.
The following theorem shows that, in a Randomized Controlled Trial (RCT) setting, pairing approaches naturally lead to consistent estimators for the estimands $\tau=\tau_\star$ and $\tau=\tau_\mathrm{pop}$.
For a Nearest Neighbor pairing choice as in \Cref{eq:pair_nn}, averaging results obtained from all comparisons lead to estimators of $\tau_\star$ that write as
\begin{equation}\label{eq:estim_NN_01}
    \hat \tau_{NN}^{(1)} = \frac{1}{|\cN_1|} \sum_{i\in\cN_1} w(Y_i|Y_{\sigma^\star_1(i)})\,,\quad \quad \hat \tau_{NN}^{(0)} = \frac{1}{|\cN_0|} \sum_{j\in\cN_0} w(Y_{\sigma^\star_0(j)}|Y_j)\,,
\end{equation}
or as 
\begin{equation}\label{eq:estim_NN}
    \hat\tau_\mathrm{NN} = \frac{1}{n}\sum_{i=1}^n T_i w(Y_i|Y_{\sigma_1^\star(i)}) + (1-T_i) w(Y_{\sigma_0^\star}(i)|Y_i)\,,
\end{equation}
where $\sigma_t^\star : \cN_t\to \cN_{1-t}$ assign the closest element from the other group, as defined in \Cref{eq:pair_nn}.
Instead, for a complete pairing choice $\cC=\cN_1\times\cN_0$, averaging all comparisons leads to a consistent estimator of $\tau_\mathrm{pop}$ that writes as:
\begin{equation}\label{eq:estim_pair_complete}
    \hat \tau_\mathrm{Tot} = \frac{1}{|\cN_1||\cN_0|}\sum_{(i,j)\in\cN_1\times\cN_0}w(Y_i|Y_j)\,.
\end{equation}
Note that the complete pairing part of the following theorem is already a well-known result \citep{Mao2017}, as the corresponding estimator is a classical $U-$statistics. We still include its proof for completeness.

\begin{theorem}[Consistency of Win Ratio, RCT]\label{thm:consistency_WR}
	Assume that \Cref{hyp:sutva,hyp:unconfoundedness,hyp:positivity} hold, and assume further that we are in the RCT setting: $T_i\indep X_i$.
	\begin{enumerate}
		\item \label{thm:consistency_WR_Complete} \emph{Complete pairing.}
	We have:
	\begin{equation*}
		\hat \tau_\mathrm{Tot} \underset{\P}{\longrightarrow} \tau_\mathrm{pop}= \esp{w(Y_i(1)|Y_j(0))}\,,\quad i\ne j\,,
	\end{equation*}
	where the limit in probability is taken as $n_0,n_1\to \infty$ and $\cC_\mathrm{Tot}=\cN_0\times \cN_1$.

		\item \label{thm:consistency_WR_NN} \emph{Nearest Neighbor.} Assume that $(x,y)\mapsto \esp{w(Y_i(1)|y)|X_i=x}$ is continuous in its first variable $x$ 
		and that
		$\cX$ is compact.
        Let $\hat\tau$ be defined as either $\hat\tau_\mathrm{NN}$ (\Cref{eq:estim_NN}) or $\hat\tau^{(t)}$ for $t\in\set{0,1}$ (\Cref{eq:estim_NN_01}).
        We have:
		\begin{equation*}
			\hat \tau\underset{\P}{\longrightarrow} \tau_\star = \esp{\esp{w(Y^{(X_i)}(1)|Y_i(0))|X_i}}\,,
		\end{equation*}
		where the limit in probability is taken as $n_0,n_1\to \infty$.

	\end{enumerate}
	
\end{theorem}

The consistency results of \Cref{thm:consistency_WR} formalize the intuition brought by \Cref{ex:counter_emp,remark:pop_vs_indiv} and illustrated in \Cref{fig:comp_WR_all_NN_adv}.
The choice of pairing $\cC$ is crucial, and leads to the estimation of very different quantities and may lead to different 
treatment recommendations.
For complete pairings $\cC_\mathrm{Tot}$, the causal estimand that is estimated is the population-level one, $\tau_\mathrm{pop}$.
For Nearest Neighbor pairings $\cC_\mathrm{NN}$, the causal estimand that is estimated is the individual-level one that we introduced, $\tau_\star$, defined in \Cref{def:win_indiv}.

As formalized in \Cref{rem:strata}, forming risk strata may be an interesting strategy, that lies in-between the two extreme choices $\cC_\mathrm{Tot}$ and $\cC_\mathrm{NN}$ for pairing sets.
The causal estimand related to a strata function can also be defined, and we expect results similar to those of \Cref{thm:consistency_WR} to hold.
We argue that the individual-level causal estimand $\tau_\star$ can in fact be recovered using infinitesimal stratas, further justifying the strength of \Cref{def:win_indiv}.

\begin{remark}[Win Ratio and comparisons with strata]\label{rem:strata}
	The following notion of strata could be developed further, as an in-between the extreme cases presented in \Cref{thm:consistency_WR}.
	Let $r:\cX\to\cR$ be a risk strata function, on some metric space $(\cR,d)$. A classical risk strata function is the \textit{baseline risk} $r(x)=\esp{Y_i(0)|X_i=x}$.
	Using this strata, the following pairs of treated and control patients:
	\begin{equation*}
		\cC_{\hat r,\eps} = \set{(i,j)\in\cN_0\times \cN_1\,|\, d(r(X_i),r(X_j))\leq \eps}\,,
	\end{equation*}
	to obtain the estimator $\hat p_\mathrm{W}$ (\Cref{eq:win_prop}) using these pairs.
	Now, define the $\eps-$strata causal estimand as: 
	\begin{equation*}
		\tau_\star^{r,\eps}= \proba{Y_i(1)\succ Y_j(0)|d(r(X_i),r(X_j))\leq \eps}\,,
	\end{equation*}
	where $i$ and $j$ are two different and independent indices.
	Under adequate assumptions, if $\cC=\cC_{\hat r,\eps}$ we expect to have that:
	\begin{equation*}
		\hat\tau_\mathrm{strata}({\hat r,\eps}) \underset{\P}{\longrightarrow}\tau_\star^{r,\eps}\,,
	\end{equation*}
    where
    \begin{equation*}
        \hat\tau_\mathrm{strata}({\hat r,\eps})= \frac{1}{|\cC_{\hat r,\eps}|}\sum_{(i,j)\in\cC_{\hat r,\eps}} w(Y_i|Y_j)\,.
    \end{equation*}
	Furthermore, our causal measure $\tau_\star$ can in fact be seen as the limit of $\tau^{r,\eps}_\star$ for $\eps\to0$.
	However, this limit may not always be well-defined, hence the strength of \Cref{def:win_indiv}.
\end{remark}

It is worth noting that using Nearest Neighbor pairings to estimate $\tau_\star$ also has limitations.
First, the method is confined to the \emph{RCT} setting: the consistency result in \Cref{thm:consistency_WR}.\ref{thm:consistency_WR_NN} relies on the independence of $T_i$ and $X_i$. A natural extension is therefore to move beyond the RCT framework. In \Cref{sec:obs_nn}, we show—following \citep{abadie2006large}—that nearest-neighbor matching extends naturally to observational data.
Second, the approach suffers from the curse of dimensionality: its convergence rate deteriorates exponentially with the dimension of the covariates $X_i$, a well-documented weakness of nearest-neighbor methods \citep{biau2015lectures}. Although extensively studied, these methods become ineffective in high dimensions unless $n$ grows exponentially.
These limitations motivate the search for a more general and robust procedure to estimate $\tau_\star$. In \Cref{sec:distributional_regression}, we adopt a distributional-regression perspective and develop more scalable methodologies for estimating $\tau_\star$, which also accommodate missing covariate values—unlike nearest-neighbor approaches.

\subsection{The matching approach in observational settings}
\label{sec:obs_nn}

We now extend our consistency results to the observational setting.
A natural approach is to resort to \emph{Inverse Propensity Weighting} (IPW) \citep{robins1994estimation,horvitz1952generalization}.
Recall that for $x\in\cX$, the propensity score $\pi(x)=\proba{T_i=1,|,X_i=x}$ (see \Cref{hyp:positivity}) gives the conditional probability of treatment. Given an estimate $\hat\pi$ of $\pi$, 
 we adapt the nearest-neighbor estimator by incorporating inverse propensity weights, as detailed in \Cref{app:ipw} to recover consistency under \Cref{hyp:unconfoundedness}, rather than the stronger assumption $T_i!\indep! X_i$.

However, noting that $\tau_\mathrm{NN}^{(t)}$ for $t=0$ and $t=1$ are biased in a known way in that they estimate, respectively, average effects on the control and treated, the nearest neighbor approach actually does not need to be reweighted, using:
\begin{equation*}
    \esp{\tau_\star(X_i)} = \proba{T_i=1} \esp{\tau_\star(X_i)|T_i=1} + \proba{T_i=0} \esp{\tau_\star(X_i)|T_i=0}\,,
\end{equation*}
leading to the consistency of $\hat\tau_\mathrm{NN}$ defined in \Cref{eq:estim_NN}.
The following theorem establishes consistency of $\hat\tau_\mathrm{NN}$ in the observational setting, using the fact that $\hat\tau_\mathrm{NN}=\frac{|\cN_1|}{n}\hat\tau_\mathrm{NN}^{(1)} + \frac{|\cN_0|}{n}\hat\tau_\mathrm{NN}^{(0)} $.
It is a generalization of popular matching methods \citep{abadie2006large} that naturally balance treated and control populations without need for nuisance parameter estimation, beyond the ATE with the risk difference that they consider.

\begin{theorem}[Nearest neighbors, observational]\label{thm:consistency_WR_observational}
Assume that \Cref{hyp:sutva,hyp:positivity,hyp:unconfoundedness} hold (SUTVA, positivity and no unobserved confounders), $\cX$ is compact, and that $(x,y)\mapsto \esp{w(Y_i(1)|y)|X_i=x}$ and $(x,y)\mapsto \esp{w(y|Y_i(0))|X_i=x}$ are continuous in $x$.
Then, for both $t=0$ and $t=1$, we have:
\begin{equation*}
    \hat\tau_\mathrm{NN}^{(t)} \underset{\P}{\longrightarrow} \esp{\tau_\star(X_i)|T_i=t}\,,
\end{equation*}
and
\begin{equation*}
    \hat\tau_\mathrm{NN} \underset{\P}{\longrightarrow} \tau_\star\,.
\end{equation*}
\end{theorem}

Note that this result is quite strong: a simple pairing approach works beyond RCT, as traditionally done in hierarchical outcome Win Ratio analysis, in contrast with inverse propensity weighting estimators proposed by \citep{Mao2017} and by many subsequent works \citep{Chen2024,Chiaruttini2024,Guo2022,Yin2022,zhang2022causalinferencewinratio}. However, these works estimate the population-level estimand $\tau_\mathrm{pop}$, that indeed requires IPW approaches to reweight pairs.

The matching estimator in \Cref{eq:estim_NN} can also be framed as a simple plug-in estimator.
Indeed, $\tau_\star$ writes as $\tau_\star=\esp{\tau_\star(X_i)}$.
Thus, an intuitive approach to estimate $\tau_\star$ would be to estimate the nuisance parameter function $x\mapsto \tau_\star(x)$, where $\tau_\star(x)$ is the win probability of a treated patient with features $x$, over a control patient with same features $x$.
Denoting as $\hat p_1$ an estimator of the nuisance function $x\mapsto\tau_\star(x)$, a plug-in estimator would write as:
\begin{equation*}
    \tau_\mathrm{plug-in}=\frac{1}{n}\sum_{i=1}^n \hat p_1(X_i)\,.
\end{equation*}
Then, as long as $\esp{\hat p_1(X_i)|X_i=x}\to \tau_\star(x)$, this estimator would be consistent.
This is exactly what we proved in \Cref{thm:consistency_WR_observational}, for 
\begin{equation}\label{eq:p_1}
    \hat p_1(X_i) = T_i w(Y_i|Y_{\sigma_1(X_i)}) + (1-T_i)w(Y_{\sigma_0(X_i)}|Y_i)\,.
\end{equation}
Note that this nearest neighbor estimator is not pointwise consistent. However, as just mentioned, we have $\esp{\hat p_1(X_i)|X_i=x}\to p_1(x)$ as $n$ grows which is sufficient for the plug-in estimator to be consistent. To have a pointwise consistent estimator (without the conditional mean), we would instead require to use $k-$nearest neighbors, with $1\lll k \lll \log(n)$. However, increasing $k$ can lead to much larger asymptotic biases \citep{abadie2006large}, hence our choice not to do so. Another approach to estimate the nuisance function $x\mapsto \tau_\star(x)$ would also be to train a (non-)parametric regressor to learn $(x,x')\mapsto \esp{w(Y_i|Y_j)|X_i=x,X_j=x',T_i=1,T_j=0}$, using samples $(a_e,b_e)$ for $e\in\cN_1\times \cN_0$, $a_{(i,j)}=(X_i,X_j)$ and $b_{(i,j)}=w(Y_i|Y_j)$ the corresponding label. We do not pursue this approach here, since we do not require pointwise consistency; it has however been considered by \citep{petit2025optimal} in the context of policy learning.

The strength of \Cref{thm:consistency_WR_observational} and of this nearest neighbor approach is that the same estimator can be used in both RCT and observational settings, pushing further towards hierarchical outcome and Win Ratio analyzes on real-world observational data, by highlighting the robustness of our matching approach.

\section{Estimation \textit{via} Nuisance Parameters Estimation}\label{sec:nuisance}

In this section, we introduce new methods to approximate $\tau_\star$ beyond the matching approach previously considered.
Our approach significantly departs from the existing Win Ratio literature, in that \textit{(i)} different nuisance parameters than the one involved in the definition of $\tau_\star$ (\Cref{def:win_indiv}) are introduced, and \textit{(ii)} we use Machine Learning tools to approach this nuisance parameter, instead of matching approaches.
This estimation involves performing distributional regression, which we do using distributional random forests. The strength of this new proposed plug-in distributional regression approach (\Cref{sec:distributional_regression}) lies in the fact that it can directly handle missing values in the covariates and bypasses the slowly converging non-asymptotic bias of matching approaches.

Tackling the estimation of $\tau_\star$ by first estimating nuisance parameter functions naturally leads to consider the semi-parametric efficiency of the approach. We thus introduce in \Cref{sec:EIF} a 1-step corrected estimator of the plug-in approach, that corrects the plug-in bias of the matching approach and that comes from computing the efficient influence function (EIF) of the estimand $\tau_\star$. This estimator inherits the properties of the semi-parametric literature (minimal variance amongst a certain class of estimators).

\Cref{sec:distributional_regression_ex} finally exposes how we perform the estimation of the new nuisance parameters we introduce next, via distributional random forests \citep{cevid2022distributional}.

\subsection{Distributional Regression Approach}\label{sec:distributional_regression}

The distributional regression approach is the counterpart of the \emph{Two-learners} or \textit{Plug-in G-formula} approaches, used for (C)ATE estimation with the risk difference.
When estimating 
\[\tau_\mathrm{RD}=\esp{\esp{Y_i|X_i,T_i=1}-\esp{Y_i|X_i,T_i=0}}\,,\] 
with a regression approach, the idea is to learn two (non)parametric estimates
$\hat \mu_t:\cX\to \R$ of $\mu_t:x\mapsto \esp{Y_i|T_i=t,X_i=x}$,
to approximate $\tau_\mathrm{RD}$ with: 
\[\frac{1}{n}\sum_{i=1}^n\hat\mu_1(X_i)-\hat\mu_0(X_i)\,.\]
A naive adaptation of this regression approach to our problem of estimating $\tau_\star$ would thus be to first learn $x\mapsto \hat \mu_t(x)$ (non)parametric estimates of $\mu_t:x\mapsto\esp{Y_i|X_i=x,T_i=t}$ as before.
Then, let $\hat p(X_i) = w(\hat y_1(X_i)|\hat y_0(X_i))$, that aims at estimating $\esp{w(Y_i(1)| Y_i(0))|X_i}$, to obtain the estimator $\frac{1}{n}\sum_{i=1}^n \hat p(X_i)$.
However, as opposed to (C)ATE estimation, we don't have linearity of $w$ here, so that even if the conditional expectations are perfectly estimated ($\hat \mu_t=\mu_t$), we won't even have consistency.
Regressing the conditional expectations makes us loose information on the way.

Hence the \emph{ distributional regression approach}, since we need to go beyond learning conditional expectations.
We start from the identification formula of $\tau_\star(x)$ (\Cref{eq:id_tau_star}). We have, by considering a treated individual:
\begin{align*}
    \tau_\star(x) &= \int w(y|y') \dd\P_\mathrm{obs}(Y_i=y|X_i=x,T_i=1)\dd\P_\mathrm{obs}(Y_i=y'|X_i=x,T_i=0)\\
    &= \esp{\int w(Y_i|y') \dd\P_\mathrm{obs}(Y_i=y'|X_i=x,T_i=0)  \Big|X_i=x,T_i=1}\\
     &= \int \esp{w(Y_i|y')\Big|X_i=x,T_i=1} \dd\P_\mathrm{obs}(Y_i=y'|X_i=x,T_i=0)\\
     &\qquad\qquad\qquad \text{(switching conditional mean and integration)}  \\
     &= \esp{q_0(x,Y_j)|T_j=0,X_j=x}\,,
\end{align*}
where 
\begin{equation*}
    q_0(x,y)=\esp{w(y|Y_i)|T_i=0,X_i=x}\,.
\end{equation*}
Similarly, $\tau_\star(x)=\esp{q_1(Y_i,x)|T_i=1,X_i=x}$ where:
\begin{equation*}
    q_1(x,y)=\esp{w(Y_i| y)|T_i=1,X_i=x}\,,
\end{equation*}
so that $\tau_\star(X_i) = \esp{T_i q_1(X_i,Y_i)+(1-T_i)q_0(X_i,Y_i)\Big|X_i,T_i}$, leading to:
\begin{equation*}
    \tau_\star=\esp{T_i q_1(X_i,Y_i)+(1-T_i)q_0(X_i,Y_i)}\,.
\end{equation*}
A natural plug-in estimator thus consists in estimating $\hat q_1,\hat q_0$ estimators of the nuisance parameters $q_1$ and $q_0$, learned on an independent dataset or split.
Let thus:
\begin{equation}\label{eq:DRF}
	\hat \tau_\mathrm{reg}= \frac{1}{n}\sum_{i=1}^n (1-T_i)\hat q_0(X_i,Y_i) + T_i\hat q_1(X_i,Y_i)\,,
\end{equation}
be the \textit{distributional regression estimator}.
Provided that the estimation errors tend to zero (in mean over the population), we prove in the next Theorem that $\hat\tau_\mathrm{reg}$ is consistent, and even asymptotic normal if nuisance parameters are sufficiently well estimated.

\begin{theorem}[Consistency and asymptotic normality]
	\label{thm:WR_DRF}
	Assume that \Cref{hyp:sutva,hyp:unconfoundedness,hyp:positivity} hold.
	We have the followings.
	\begin{enumerate}
		\item \label{thm:WR_DRF_consistency}
		Assume that for $t\in\set{0,1}$, we have:
		\begin{equation*}
			\esp{|\hat q_t(X_i,Y_i)-q_t(X_i,Y_i)|\,\big |\, T_i=1-t}\to 0\,.
		\end{equation*}
		Then, the estimator $\hat\tau_\mathrm{reg}$ defined in \Cref{eq:DRF} converges almost surely to $\tau_\star$ (defined in \Cref{def:win_indiv}).
		\item \label{thm:WR_DRF_normal}
		 Assume that for $t\in\set{0,1}$, we have:
		\begin{equation*}
			\esp{|\hat q_t(X_i,Y_i)-q_t(X_i,Y_i)|\,\big |\, T_i=1-t}=o\big(n^{-1/2}\big)\,.
		\end{equation*}
		Then, $\hat\tau_\mathrm{reg}$ satisfies:
		\begin{equation*}
			\sqrt{n}\big(\hat \tau_\mathrm{reg}-\tau_\star  \big)\underset{\P}{\longrightarrow} \cN(0,\sigma_\infty^2)\,,
		\end{equation*} 
		where
\begin{align*}
    \sigma_\infty^2&=\proba{T_i=1}\var(\esp{w(Y_i(1), Y_i(0))|X_i,Y_i(0),T_i=1})\\
    &\quad+\proba{T_i=0}\var(\esp{w(Y_i(1), Y_i(0))|X_i,Y_i(1),T_i=0})
\end{align*}
	    is the variance of the probability of a win conditioned on covariates, treatment, and counterfactual.
	\end{enumerate}
\end{theorem}

The assumption of \Cref{thm:WR_DRF}.\ref{thm:WR_DRF_consistency} will hold, as long as the distributional regressors are consistent. This will for instance be the case of Distributional Random Forests \citep{cevid2022distributional, pmlr-v238-benard24a_DRF} that we use (defined and explained further in \Cref{sec:distributional_regression_ex}), under very mild assumptions.
The assumption of \Cref{thm:WR_DRF}.\ref{thm:WR_DRF_normal} is much stronger, and requires a fast parametric rate of convergence. It will hold if for instance we perform logistic regression (\Cref{sec:distributional_regression_ex}) for a well-specified distributional regression problem.
Under such assumptions, the asymptotic normality yields asymptotically valid confidence intervals.

\subsection{Semi-parametric efficient estimator}\label{sec:EIF}

When plugging-in nuisance functions estimators as in \Cref{eq:DRF}, if the models have large non-asymptotic biases (as for nearest-neighbors), if they overfit (as for forests) or if they are mispecified (as for parametric models), this can lead to some bias in the estimated value.
Efficient Influence Functions (EIFs) can be used to correct for this plug-in bias \citep{hines2022demystifying}, and to obtain efficient estimators that satisfy, under mild assumptions, double robustness properties, semi-parametric efficiency optimality and minimal asymptotic variance amongst certain class of estimators.

Let $\Psi(\cP)$ be an estimand we wish to estimate, that is the function of an observed distribution $\cP$.
In our case, $\cP$ is the distribution of our $n$ observations $O_i=(X_i,T_i,Y_i)$, and 
\begin{equation*}
    \Psi(\cP) = \esp{\tau_\star(X_i)} = \esp{\esp{w(Y^{(X_i)}(1)|Y_i) | X_i,T_i=0}}\,,
\end{equation*}
as defined in \Cref{def:win_indiv}.
As we saw in the previous subsection, we can estimate $\Psi(\cP)$ by learning (non-parametric) nuisance estimates (with distributional regression here) on the data.
However, small perturbations and errors in these nuisance estimations can then lead to large biases in our estimate of $\Psi$ ($\tau_\star$ here).
The \emph{plug-in} estimator defined in the previous subsections writes as:
\begin{equation*}
	\Psi_\mathrm{plug-in} = \Psi(\hat\cP_n)\,,
\end{equation*}
where $\hat\cP_n$ is an estimate of the distribution $\cP$ on $n$ samples (empirical distribution, and estimation of nuisance parameters).
Denoting as $\phi(\cP)$ the Efficient Influence Function of $\Psi$ at $\cD$, the 1-step estimator corrects first-order biases that come from wrongly estimating the nuisance parameters of the distribution:
\begin{equation*}
	\Psi_\mathrm{1-step}= \Psi_\mathrm{plug-in} + \frac{1}{n}\sum_{i=1}^n \Phi(\hat \cP_n)(O_i)\,.
\end{equation*}
We thus need to compute the EIF $\phi(\cP)(o)$ of $\Psi(\cP)=\tau_\star$ for any distribution $\cP$ and observation $o$.

\begin{proposition}[EIF]
\label{prop:EIF}
	The EIF of $\Psi$ writes as, for $o=(X,T,Y)$:
		\begin{align*}
			\phi(\cP)(o) &= -\Psi(\cP) + \esp{w(Y_i(1)|Y^{(X_i)}(0))|X_i=X}\\
			& \quad + \frac{T}{\pi(X)} \left( \esp{w(Y|Y^{(X)}(0))} - \esp{w(Y^{(X_i)}(1)|Y^{(X_i)}(0))|X_i=X} \right)\\
			& \quad + \frac{1-T}{1-\pi(X)} \left( \esp{w(Y^{(X)}(1)|Y)} - \esp{w(Y^{(X_i)}(1)|Y^{(X_i)}(0))|X_i=X} \right) \,.
		\end{align*}
\end{proposition}

Let $\hat \pi$ be an estimator of the propensity score, $\hat q_1,\hat q_0$ as in \Cref{sec:distributional_regression} be the distribution regression estimators, and $\hat p_1$ be an estimator of the nuisance parameter $x\mapsto \tau_\star(x)$ (win probability conditioned on covariates).
The efficient influence function of $\psi$ then writes as, for a distribution $\cP$ and observation $o=(X,T,Y)$:
\begin{equation}\label{eq:EIF}
	\Phi(\cP)(o) = -\Psi(\cP) +  p_1(X) + \frac{T}{ \pi(X)}\left(q_1(X,Y)-p_1(X)) \right) + \frac{1-T}{1- \pi(X)}\left(q_0(X,Y)-p_1(X)\right)\,.
\end{equation}

\paragraph{1-step estimator.}
The 1-step estimator $ \Psi_\mathrm{plug-in} + \frac{1}{n}\sum_{i=1}^n \Phi(\hat \cP_n)(O_i)$ then writes as:
\begin{equation}\label{eq:1-step}
\begin{aligned}
	\hat\tau_\mathrm{1-step} & = \frac{1}{n}\sum_{i=1}^n \Big[ \hat p_1(X_i) + \frac{T_i}{\hat \pi(X_i)}\left(\hat q_1(X_i,Y_i)- \hat p_1(X_i)\right) + \frac{1-T_i}{1- \hat\pi(X_i)}\left(\hat q_0(X_i,Y_i)-\hat p_1(X_i)\right) \Big]\,,
\end{aligned}
\end{equation}
obtained by plugging in nuisance estimators in \Cref{eq:EIF} for $p_1,q_0,q_1$.
For $\hat p_1$ as in \Cref{eq:p_1} (nearest neighbor), this leads to:
\begin{equation*}
\begin{aligned}
	\hat\tau_\mathrm{1-step} & = \frac{1}{n}\sum_{i=1}^n \Big[ T_i w(Y_i|Y_{\sigma_1(X_i)}) + (1-T_i)w(Y_{\sigma_0(X_i)}|Y_i)\\
    &\quad + \frac{T_i}{\hat \pi(X_i)}\left(\hat q_1(X_i,Y_i)- w(Y_i|Y_{\sigma_1(X_i)})\right) + \frac{1-T_i}{1- \hat\pi(X_i)}\left(\hat q_0(X_i,Y_i)-w(Y_{\sigma_0(X_i)}|Y_i)\right) \Big]\,.
\end{aligned}
\end{equation*}
Note that influence functions-based estimators can also be built using \textit{estimating equations}. However, in our setting, since the EIF writes as $\phi(\cP)(o) = -\Psi(\cP) + h(\eta,o)$ where $\eta$ are the nuisance parameters, both estimating equations and 1-Step estimators are \textit{equal}.
This is typically the case for the EIF augmented estimators of the ATE with the RD that both yield the AIPW estimator.

Finally, we describe in \Cref{sec:distributional_regression_ex} methods to perform distributional regression, to learn nuisance parameters $q_1$ and $q_0$.
We propose both a non-parametric distributional regression method, and a linear-parameterized one.

\section{Experiments on randomly generated data
}






We generate synthetic observational data as follows.
\begin{enumerate}
    \item Covariates are generated as standard multivariate Gaussian random variables $X_i\sim\cN(0,I_p)$.
	\item Treatment assignments follow a Bernoulli law of mean $\sigma(\langle X_i,v\rangle)$, where $v\in\R^p$ is unitary and $\sigma:\R\to[0,1]$ is the Gaussian cdf.  We took $v=(0,\ldots,0,1)$.
	\item Potential outcomes $Y_i(0),Y_i(1)\in\set{0,1}^d$ follow multidimensional Bernoulli laws of parameters $\langle u_k^{(t)},X_i\rangle,k\in[d]$ for $u_k^{(t)}\in\R^p$:
	\begin{equation*}
		\proba{Y_i(t)=(e_1,\ldots,e_d)}=\prod_{k=1}^{d} \Big(e_k \sigma(\langle u_k^{(t)},X_i\rangle) + (1-e_k)(1-\sigma(\langle u_k^{(t)},X_i\rangle))\Big)\,,
	\end{equation*}
	for any $(e_1,\ldots,e_d)\in\set{0,1}^d$.
    For $Y_i(0)$, we choose $u_k^{(0)}=0$, so that coordinates of $Y_i(0)$ are independent Bernoulli random variables of parameter $0.5$.
\end{enumerate}
We then distinguish two scenarios in terms of outcome generation, and dimensions of covariates and outcomes.
\begin{enumerate}
    \item \textbf{\Cref{fig:dimension_outcomes}:} we study how increasing the number of outcomes complexifies the task.
    Outcomes are generated by using $u^{(1)}$ a unitary vector, and setting $u_k^{(1)}=u^{(1)}$ for all $k\in[d]$. We choose $u^{(1)}=(1,\ldots,1)$.
    Feature dimension is fixed and small ($p=3$), while outcome dimension varies ($d\in\set{3,8,15}$).
    \item \textbf{\Cref{fig:doubly_robust_mis_dr}}: we study how incresing the dimension of covariates and having mispecified distribution regression affects our estimators, with $u_1^{(1)}=0$ and $u_k^{(1)}=(0,\ldots,0,1)$ for $k\geq 2$. 
    Outcome dimensions is fixed ($d=3$) while feature dimension varies ($p\in\set{3,8,15}$).
\end{enumerate}
Violin plots are made over 200 runs, forests are taken with 1000 trees.
In all plots, the red line corresponds to the ground truth estimand (\Cref{def:win_indiv}).
We compare the following estimators:
\begin{enumerate}
    \item \textit{NN} corresponds to the nearest neighbor estimator (\Cref{eq:estim_NN});
    \item 
    \textit{Distributional regression methods.}
    DRF corresponds to the distributional regression approach (\Cref{eq:DRF}), with distributional random forests to perform distributional regression, while DLR uses a distributional linear regression  that imposes a model with $u^{(k)}=u^{(1)}$ for all outcomes $k\in[d]$;
    \item \textit{EIF-augmented approaches.}
    EIF corresponds to the 1-step EIF estimator (\Cref{eq:EIF}), with distributional random forests for distributional regression, and probability forests for propensity scores. EIF Oracle corresponds to the same estimator, but with oracle propensity scores, and EIF lin uses logistic regression for propensity scores and a linear distributional regression that imposes a model with $u^{(k)}=u^{(1)}$ for all outcomes $k\in[d]$. 
\end{enumerate}
EIF lin and DLR are mispecified in \Cref{fig:doubly_robust_mis_dr}.

\textbf{Interpretation.}
Increasing the dimension of the outcomes in \Cref{fig:dimension_outcomes} (for fixed feature dimension equal to 3) makes the distributional regression task harder. However, we observe that all our methods, if well specified, are robust to this outcome dimension increase and still learn well the right win proportion, with almost no asymptotic bias for small outcome dimension of 3. For dimension 15, we expect that  the asymptotic bias would decrease for larger samples. Nearest neighbors seem to perform better in higher outcome dimensions than other methods, which is to be expected they because they are not hindered by large outcome dimensions. This difference is however subtle.

Increasing the dimension of the features in \Cref{fig:doubly_robust_mis_dr} (for fixed outcome dimension equal to 3) makes nearest neighbors matchings, propensity scores learning, and distributional regression harder.
For feature dimensions $p\in\set{3,8,15}$, all non-parametric methods (that are not mispecified)  perform well, with no or almost no non-asymptotic bias. Variance for small sample size is larger with large covariate dimension. Although one would expect some non-asymptotic biases from these non-parametric methods that work under the curse of dimensionality, this does not seem to be the case for all of them. 
Indeed, Nearest neighbor matching is non-asymptotically biased for covariate dimensions equal to 8 and 15, as well as the distributional regression methods DRF and DLR (with smaller biases for DRF, larger bias from DLR coming from mispecification).
However, our 1-step augmented approach (EIF) is non-biased even non-asymptotically, and thus appear to be more robust than matching methods, illustrating the strength of our novel approach.
Mispecified methods (DLR and EIF lin) have a large bias here.



\begin{figure}
    \centering
        \hfill
    \begin{subfigure}{0.45\textwidth}
        \centering
        \includegraphics[width=\linewidth, trim=50 150 0 50, clip]{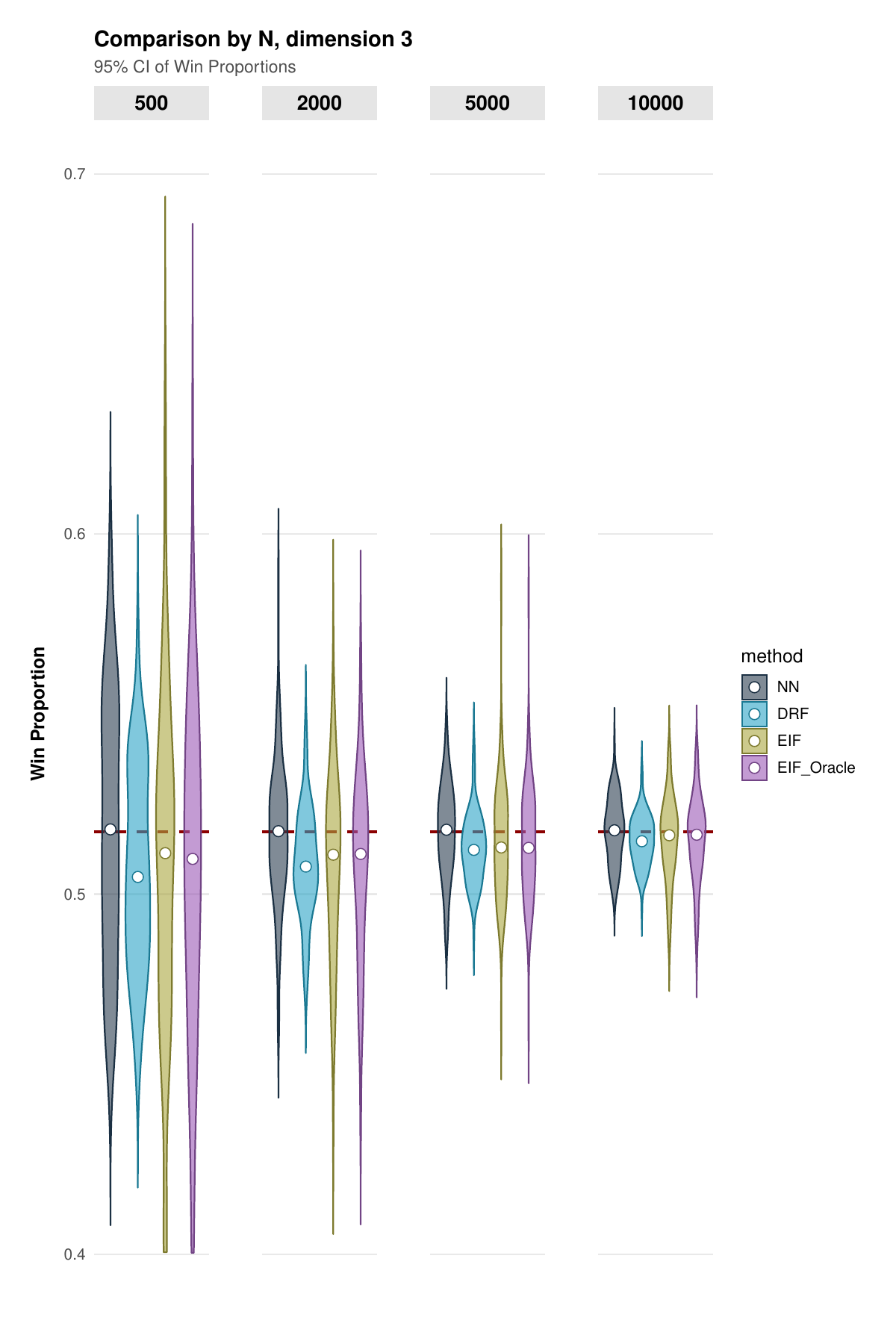}
        \caption{$d=p=3$}
        \label{fig:dimension_outcomes_3_wo_cheating}
    \end{subfigure}
    \hfill
    \begin{subfigure}{0.45\textwidth}
        \centering
        \includegraphics[width=\linewidth, trim=50 150 0 50, clip]{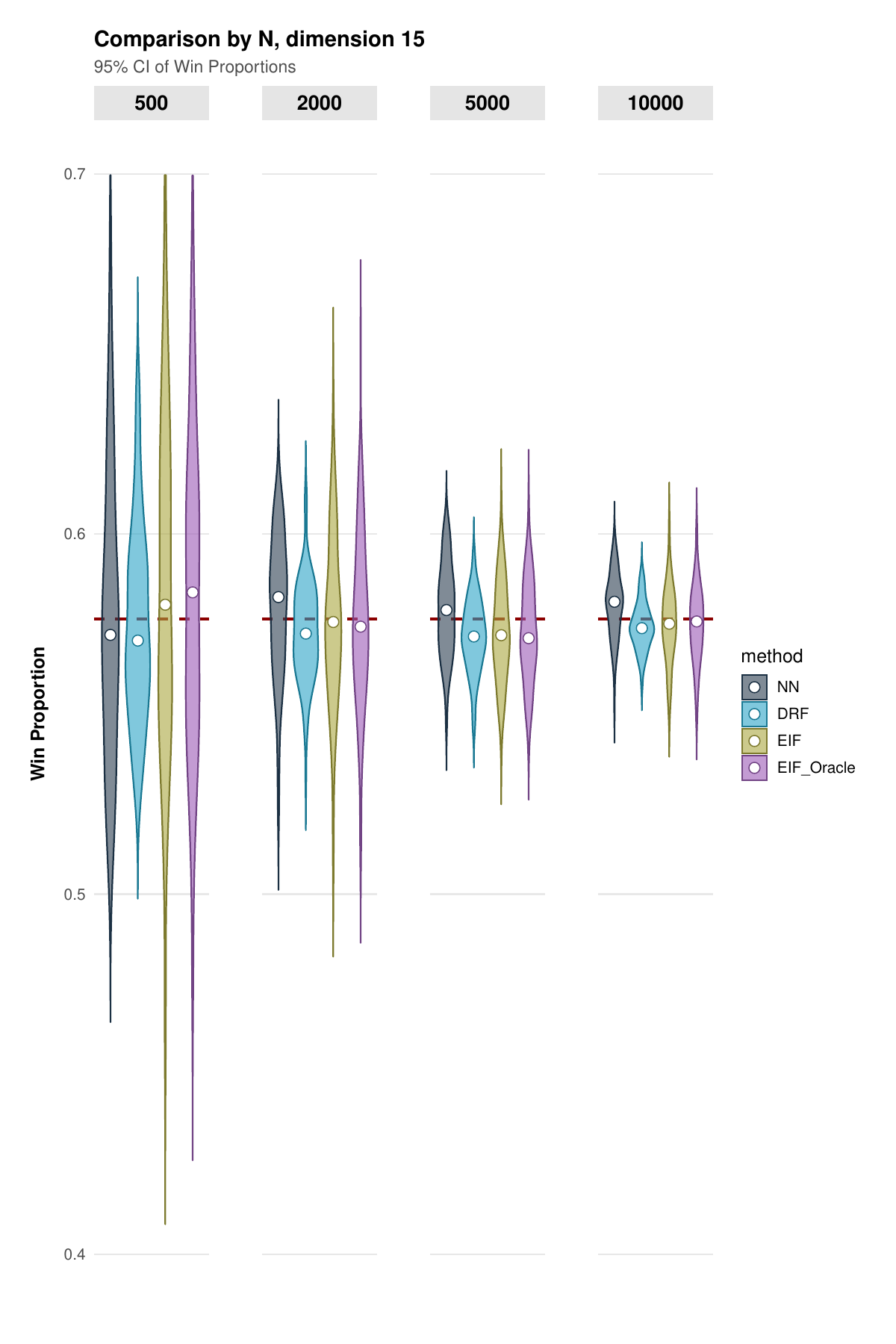}
        \caption{$d=15,p=3$}
        \label{fig:dimension_outcomes_20_wo_cheating}
    \end{subfigure}
        \hfill
    \caption{Uncorrelated outcomes setting, fixed covariate dimension ($p=3$), increasing outcome dimension ($d$).
	}
    \label{fig:dimension_outcomes}
\end{figure}

\begin{figure}
    \centering
    \begin{subfigure}{0.32\textwidth}
        \centering
        \includegraphics[width=\linewidth, trim=50 150 0 50, clip]{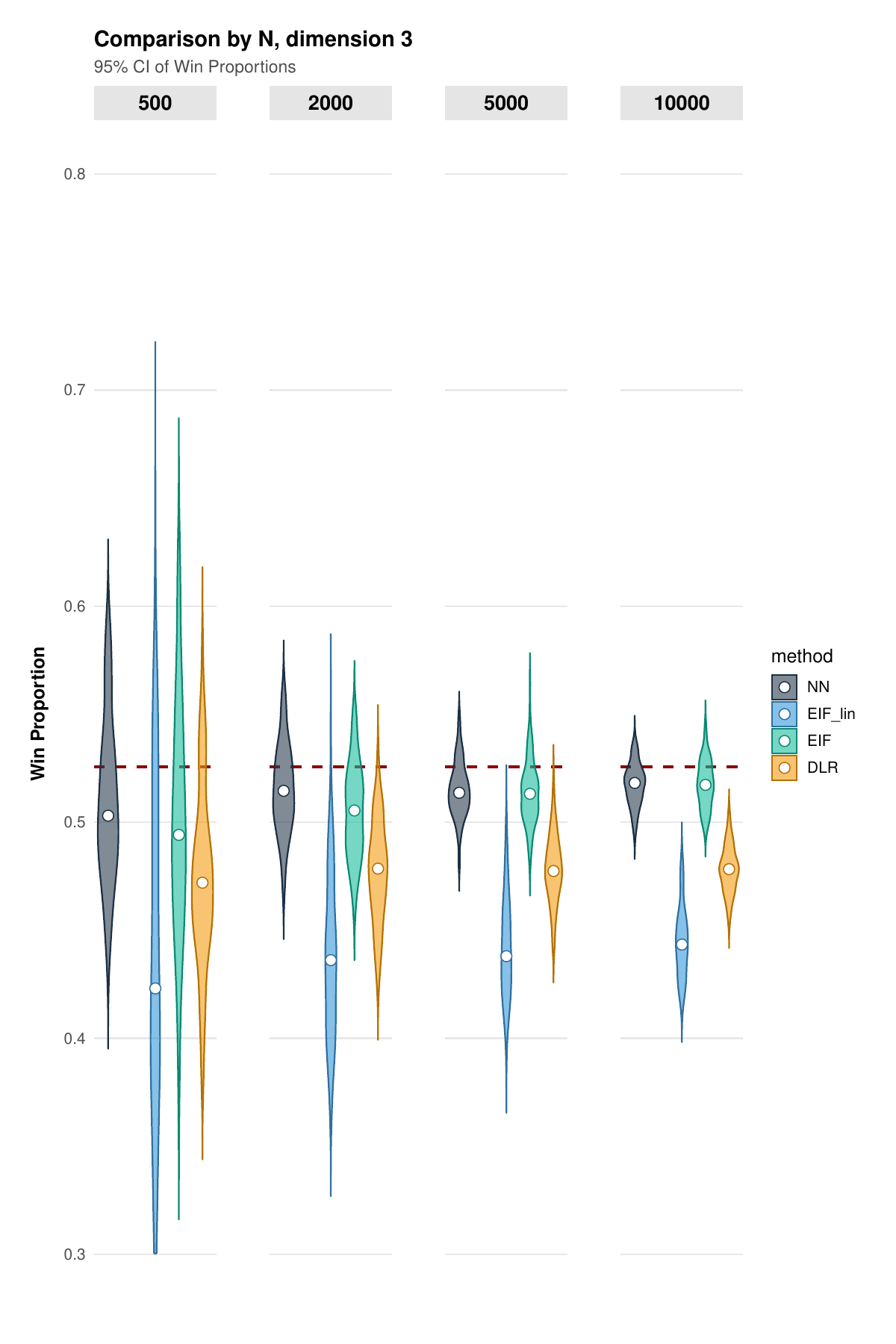}
        \caption{$d=p=3$}
        \label{fig:doubly_robust_mis_dr_3}
    \end{subfigure}
    \begin{subfigure}{0.32\textwidth}
        \centering
        \includegraphics[width=\linewidth, trim=50 150 0 50, clip]{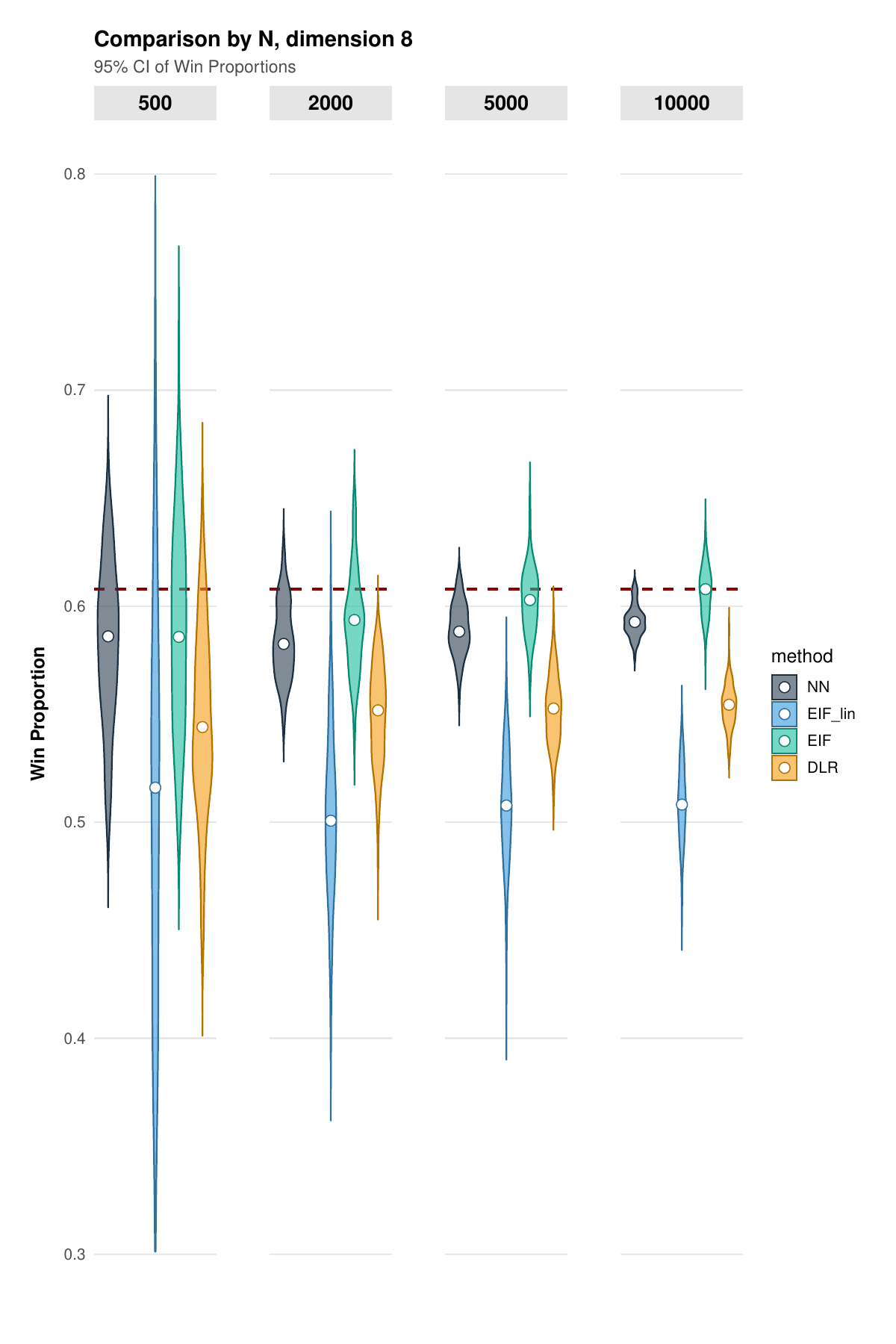}
        \caption{$d=p=8$}
        \label{fig:doubly_robust_mis_dr_8}
    \end{subfigure}
    \begin{subfigure}{0.32\textwidth}
        \centering
        \includegraphics[width=\linewidth, trim=50 150 0 50, clip]{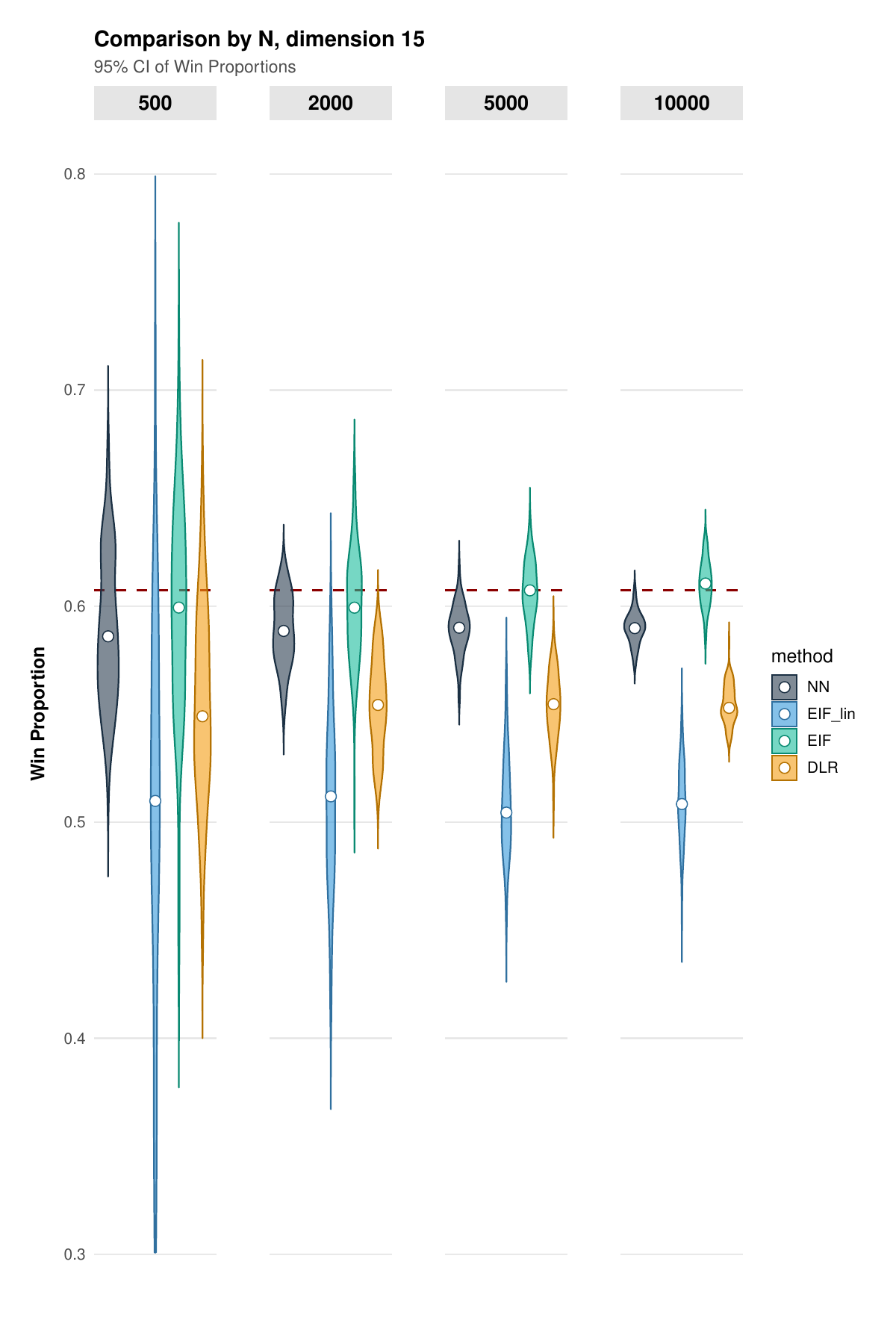}
        \caption{$d=p=15$}
        \label{fig:doubly_robust_mis_dr_15}
    \end{subfigure}

    \caption{Increasing covariate dimension, and mispecified distributional regression models for EIF lin and DLR.}
    \label{fig:doubly_robust_mis_dr}
\end{figure}

\section{Application to the CRASH-3 RCT}
\label{sec:crash3}

We illustrate our methodologies to perform a Win Ratio analysis of the CRASH-3 clinical trial \citep{crash3}, which studied the effects of tranexamic acid (\textbf{TXA}) in traumatic brain injury (TBI). 
We illustrate the different ways our methodologies can be used to derive Win Ratio estimates with confidence intervals.

\subsection{Presentation of the dataset}

\paragraph{Data description and preprocessing.}
The CRASH-3 RCT contains information on 12, 737 patients. In order to have lighter computations, we chose to use only a random sample (without replacement) of 6, 000 patients for all our analysis. 
Missing data is imputed using mice \citep{mice}.
\begin{enumerate}
    \item \textbf{Patients covariates} include: siteId (hospital identifier), sex (male/female), age (years), timeSinceInjury (hours since injury), systolicBloodPressure (mmHg), Glasgow Coma Scale scores, , pupilReact (pupil reactivity).
    \item Patients are either assigned \textbf{Placebo} (6,321) or \textbf{TXA} (6,416 patients). 
    \item \textbf{Primary outcomes} are \emph{death events} in the 28 days following trauma, that we encode as 1 or 0.
    \item \textbf{Secondary outcomes} are vascular risks. We encode them as 1 (if there is at least one stroke, heart attack, pulmonary embolism or deep vein thrombosis) or 0 (if there are no such events).
    \item \textbf{Tertiary outcomes} are the number of days the patient stayed in the hospital (censored at 28 days).
\end{enumerate}
Computing the average treatment effect for each of these 3 outcomes lead respectively to the confidence intervals, where $Y_1,Y_2,Y_3$ are respectively our death, secondary effects and hospitalization duration outcomes:
\[\esp{Y_1(1)-Y_1(0)}\in[-0.0032,0.0033]\,,\] \[\esp{Y_2(1)-Y_2(0)}\in[-0.0021,0.0079]\,,\] and \[\esp{Y_3(1)-Y_3(0)}\in[-0.2305,0.4512]\,,\] computed using the difference of means estimator. 
None of these effects are significant.

\subsection{The different methodologies used}
\label{sec:crash3methods}

For continuous features, Nearest Neighbors can be naturally applied after an eventual reweighting of the different features, to prevent from scale effects.
For categorical features, Nearest Neighbors can be applied with one-hot encodings for instance, leading to Hamming-distances.
In the presence of different categorical features of different importance, these can be weighted according to their importance.
CRASH-3 dataset has both  categorical and numerical features. We use two different ways to handle mixed data.
The first one is to consider the \textit{Mahalanobis} distance after one-hot encoding of categorical variables, that aims at naturally balancing variabilities and scale of the different features.
The second approach, that we recommand, uses a \textit{Factor Analysis of Mixed Data} (FAMD) \citep{josse2012handling,factominer}.
Since nearest neighbors algorithm relies on distance metrics (like Euclidean or Manahalahobis distances), it struggles with mixed data types, even with Manahalahobis distances that are usually used for numerical variables only. 
FAMD (the equivalent of a PCA for combined categorical and numerical covariates) transforms both numerical and categorical variables into a common latent space, ensuring a more meaningful distance calculation, i.e. balancing the influence of each variable in the computation of the principal components.
Furthermore, nearest neighbors algorithm suffers from high-dimensional data because distances become less meaningful in higher dimensions. FAMD captures the most important variations in fewer dimensions, improving NN’s effectiveness.
Finally, FAMD can be used with missing data without relying on external imputation methods, therefore extending our method to missing data in the covariates.
We compared the following methodologies in \Cref{fig:p_final}:
\begin{enumerate}
    \item Traditional Win Ratio, computed using all pairs of the dataset using the WINS package \citep{WINS}. 
    \item Stratified Win Ratio \citep{dong2018stratified}: we stratify according to the time since injury the patient received the treatment (TXA or placebo).
    \item Nearest neighbor approach for Win Ratio, as introduced in this paper, with either:
    \textbf{(i)}  the Manahalahobis distance on the covariates, or
    \textbf{(ii)} FAMD.
    We perform the FAMD analysis of the covariates, keep 95$\%$ of the variance explained, and perform nearest neighbors on these dimensions.
    Nearest neighbors are computed using the MatchIt package.
    
    \item Distributional Random Forests, as described in \Cref{sec:distributional_regression}.
    We estimate the \emph{Win Proportion} (obtained with $w(y,y')=\one_\set{y\succ y'}$) and the \emph{Loss Proportion} (obtained with $w(y,y')=\one_\set{y\preceq y'}$), and estimate the Win Ratio as the ratio between win et loss proportions.
    \item Our EIF approach, as described in \Cref{sec:distributional_regression}.
    For DRF and EIF forests are taken with 1000 trees.
\end{enumerate}
For confidence intervals, we use bootstrapping with 200 bootstraps.

\begin{figure}
        \centering
        \includegraphics[width=1\linewidth, trim=0 17 0 20, clip]{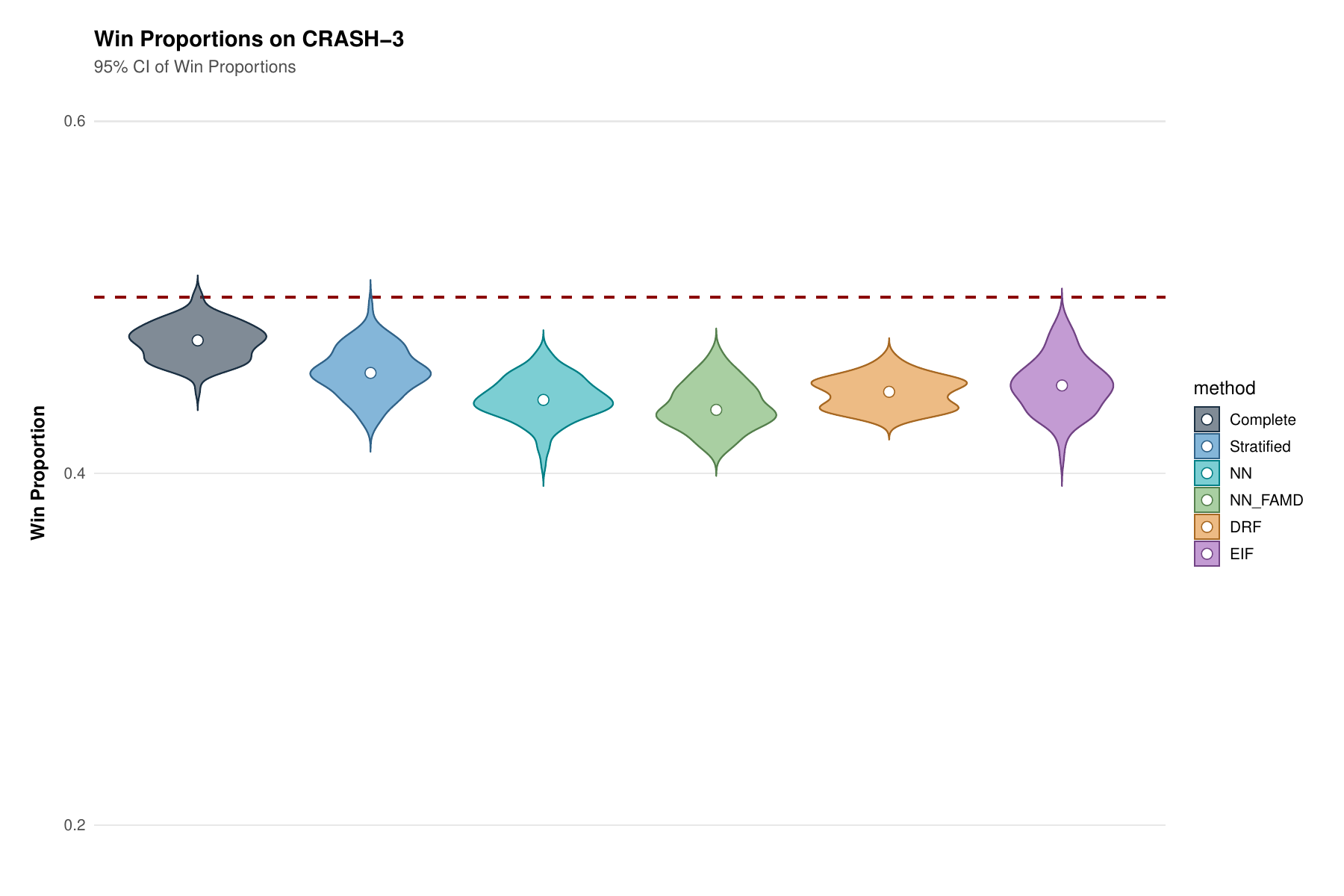}
        \caption{Win Proportion computed on the CRASH-3 dataset using different methodologies.
        }
        \label{fig:p_final}
\end{figure}

\textbf{Comments on \Cref{fig:p_final}.} 
The pairing methods (complete pairings, stratified pairings according to the time since injury, nearest neighbors) are compared in \Cref{fig:p_final}. 

\textit{Win Ratio with complete and stratified pairings.}
Complete pairing does not show significant results: 0.5 is in the confidence intervals, meaning that this method cannot conclude for or against treatment.
Stratified Win Ratio is still not significant, although stratifying seems to  drift the results towards negative treatment effects (non-significatively).

\textit{Nearest neighbors.}
Nearest neighbors obtain significant results. Indeed, matching being an extreme form of stratification, the small drift towards negative treatment effects is here magnified, and leads to have significant results.
This becomes even more apparent when pairing in the FAMD latent space.

\textit{Distributional regression approaches.}
DRF obtains similar results, but with smaller violin plot: it seems more robust than the pairing methods in this setting, illustrating the strength of our new approach.
The boxplots of our EIF augmented approaches is however a tiny bit larger.

\textit{Conclusion from this study.}
Our Win Ratio analysis of the CRASH-3 study suggests a preference over placebo rather than over treatment.
This is not in contradiction with the conclusion of the CRASH-3 study, that were in favor of treatment only \emph{on a subpopulation}. ATEs computed on the whole population are indeed not statistically significant, as showed above.
This suggests the strength of our methodologies, that can go beyond single outcomes or separately studying each outcome.



\section{Conclusion and open directions}

In this paper, we have introduced a causal inference framework for hierarchical outcome comparison methods like Win Ratio or Generalized Pairwise Comparisons.
We offered new perspectives and shed light on the different causal effect measures that may be targeted when performing a Win Ratio or related hierarchical outcome analysis.
In particular, we highlight the fact that if the population is heterogeneous, complete pairings (the historical and traditional way of forming pairs to compute the Win Ratio or the Net Benefit of a treatment) may target a population-level estimand that reverses treatment recommendations.
The new causal effect measure $\tau_\star$ we introduce in \Cref{def:win_indiv} aims at answering this by taking into account covariate effects, thus being more robust to heterogeneous population.
We stress the fact that this causal effect measure is related to stratified Win Ratio, since it can be estimated using an extreme form of stratification i.e., a Nearest Neighbors approach when forming pairs of treated-control patients.

\section*{Acknowledgements}

M. Even acknowledges funding from PEPR Santé Numérique SMATCH\footnote{\href{https://pepr-santenum.fr/2023/11/07/smatch/}{https://pepr-santenum.fr/2023/11/07/smatch/}}.
We would like to deeply thank Uri Shalit, Stefan Michiels and François Petit for very stimulating discussions in the writing of this paper, Dovid Parnas for suggesting \Cref{ex:counter_emp} and for all the very the useful feedback, and Jeffrey Näf for the help in using the DRF package.

\bibliography{biblio.bib}
\bibliographystyle{plainnat}

\newpage
\appendix




\section{Distributional random forests and distributional linear regression for estimating $q_1(x,y)$}
\label{sec:distributional_regression_ex}

So far, we have not specified the methods used to perform distributional regression to learn the  estimators $\hat q_t$ of the nuisance parameters $q_t$.
We describe in this section a non-parametric regression approach using \emph{distributional random forests} \citep{cevid2022distributional,pmlr-v238-benard24a_DRF}, and a parametric approach using linear regression or linear logistic classifiers.

\paragraph{Distributional random forests (DRF).} 
Our goal is to estimate $q_1(x,y)=\proba{w(Y_i| y)|X_i=x,T_i=1}$ with distributional random forests \citep{cevid2022distributional,pmlr-v238-benard24a_DRF}.
Let $\cH$ a Hilbert space with a kernel $k(\cdot,\cdot)$ defined on $\cY\times \cY$ (usually the Gaussian kernel).
For any probability measure $\P$ on $\cY$, let $\phi(\P)=\esp{k(Z,\cdot)}\in\cH$, for $Z\sim\P$.
Distributional random forests approximate $\mu(x,t)=\phi(\P_{Y|X=x,T=t})$ using splitting rules in the Hilbert space $\cH$, and as such are simply generalization of vanilla random forests to Hilbert spaces.
Then, \citet{cevid2022distributional} use the fact that for kernels such as the Gaussian kernel, learning kernel representations amounts to learning probability distributions.
Thus, learning with samples $(X_1,Y_1),\ldots,(X_m,Y_m)$, $\mu(x,t)$ is approximated via $\hat \mu_m(x,t)$ that takes the form:
\begin{equation*}
	\hat \mu_m(x,t)=\sum_{i=1}^m \omega_i(x,t)k(Y_i,\cdot)\in\cH\,,
\end{equation*} 
for weights $\omega_i(\cdot,\cdot)$ learnt by the forest.
The formula on the right hand side can be written as $\sum_{i=1}^m \omega_i(x,t)k(Y_i,\cdot)=\phi\left(\sum_{i=1}^m w_i(x,t)\delta_{Y_i}\right)$, where for $y\in\cY$, $\delta_y$ is a Dirac of mass 1 at point $y$.
The distribution $\P_{Y|X=x,T=t}$ is thus approximated by:
\begin{equation*}
	\hat \P^{(m)}_{Y|X=x,T=t}=\sum_{i=1}^m \omega_i(x,t)\delta_{Y_i}\,.
\end{equation*}
Hence, the probability $\proba{Y\in\cS|X=x,T=t}$ can be estimated by:
\begin{equation*}
	\hat\P^{(m)}\left(Y\in\cS|X=x,T=t\right)=\sum_{i=1}^m \omega_i(x,t) \delta_\set{Y_i\in\cS}\,.
\end{equation*}
More generally, any conditional expectation $\esp{h(Y)|X=x,T=t}$ for some measurable $h:\cY\to \R$ can be approximated by:
\begin{equation*}
	\hat\E^{(m)}\left[ h(Y)|X=x,T=t\right] = \sum_{i=1}^n \omega_i(x,t)h(Y_i)\,.
\end{equation*}
Now, we remark that the quantity that we wish to estimate writes as $q_1(x,y)=\proba{h_y(Y_i)|X_i=x,T_i=1}$, where $h_y(Y_i)=w(Y_i| y)$.
Thus, our estimate $\hat q_1(x,y)$ of $q_1(x,y)=\proba{w(Y_i| y)|X_i=x,T_i=1}$ is:
\begin{equation*}
	\hat q_1(x,y) = \sum_{i=1}^n \omega_i(x,1)w(Y_i|y)\,,
\end{equation*}
while our estimate $\hat q_0(x,y)$ of $q_0(x,y)=\proba{w(y| Y_i)|X_i=x,T_i=0}$ is:
\begin{equation*}
	\hat q_0(x,y) = \sum_{i=1}^n \omega_i(x,0)w(y|Y_i)\,.
\end{equation*}
In practice, these steps are implemented in a very concise way\footnote{that will be released soon on github}, using the following pseudo code to implement the distributional regression estimator (\Cref{eq:DRF}).
\begin{enumerate}
    \item Dataset $\cD=\set{(X_i,T_i,Y_i),i\in[n]}$ is split between a train set $\cD_\mathrm{train}=\set{(X_i,T_i,Y_i),i\in[m]}$ and an inference set $\cD_\mathrm{inference}=\set{(X_i,T_i,Y_i),m+1\leq i\leq n}$.
    \item A Distributional Random Forest (DRF) is trained on $\cD_\mathrm{train}$ to predict $Y_i$ from $(X_i,T_i)$, using the R package \citep{drf_package}, an implementation of DRFs as introduced by \citep{cevid2022distributional}.
    \item Apply the DRF to predict on the inference set $\cD_\mathrm{inference}$, to obtain the weights $\omega_i(X_j,T_j)$ for $i$ a train point and $j$ in the inference set.
    \item Compute and output:
    \begin{equation*}
        \frac{1}{n-m}\sum_{j=m+1}^n \sum_{i=1}^m \Big[ \one_\set{T_j=0} \omega_i(X_j,0) w(Y_i|Y_j)+ \one_\set{T_j=1} \omega_i(X_j,1) w(Y_j|Y_i)\Big]\,.
    \end{equation*}
\end{enumerate}


\paragraph{Linear Distributional Regression.}
We now present a parametric approach, that learns how to sample from $\P_{Y_i(t)|X_i}$ for $t\in\set{0,1}$.

First, if outcomes are real-valued ($\cY\subset\R^d$ with $d=1$), we fit a parametric model, depending on the nature of the outcome. For $t\in\set{0,1}$,
\begin{enumerate}
    \item If $\cY=\set{0,1}$ i.e., outcomes are binary, we fit a logistic regression on $(X_i,Y_i)_{i:T_i=t}$.
    To sample from $\P_{Y_i(t)|X_i}$, we then simply sample from a Bernoulli random variable of parameter given by the logistic model evaluated on $X_i$.
    \item If $\cY=\R$ with continuous outcomes, we make a Gaussian prior and fit a linear regression on $(X_i,Y_i)_{i:T_i=t}$.
    To sample from $\P_{Y_i(t)|X_i}$, we then simply sample from a Gaussian random variable of mean given by the regression model evaluated on $X_i$, and of variance the variance on the population of $Y_i$.
    \item If $\cY$ is discrete and finite, we can also learn a multi-class linear classifier, and proceed as above.
\end{enumerate}
We thus obtain approximate distributions $\hat\P_{Y_i(t)|X_i}$ from which we can sample, leading to natural estimators of $q_1(x,y)$ and $q_0(x,y)$, as:
\begin{equation*}
    \hat q_1(x,y)=\frac{1}{m}\sum_{j=1}^m w(\hat Y_j^{(1)}|y)\,,\qquad \hat q_0(x,y)=\frac{1}{m}\sum_{j=1}^m w(y|\hat Y_j^{(0)})\,,\qquad (\hat Y_j^{(t)})_{j\in[m]}\sim \hat \P_{Y_i(t)|X_i}^{\odot m}\,,
\end{equation*}
for some $m\geq 1$.

Then, for multivariate outcomes, we use the Rosenblatt transform~\citep{rosenblatt1952remarks}:
\begin{enumerate}
    \item We do as described just above to sample from the first coordinate $Y_{i1}(t)|X_i=x$ of the outcome.
    \item Then, iteratively, for $k\in\set{2,\ldots,d}$, having sampled $(Y_{i1}(t),\ldots,Y_{i,k-1}(t))$ from $\P_{(Y_{i1}(t),\ldots,Y_{i,k-1}(t))|X_i=x}$, we sample $Y_{ik}(t)|(Y_{i1}(t),\ldots,Y_{i,k-1}(t)),X_i=x$ as above, by putting the previous coordinates of the outcomes in the covariates.
\end{enumerate}
This leads to a parametric distributional regression method, that learns $\hat\P_{Y_i(t)|X_i=x}$.

\section{Proof: RCT setting and pairings}

\subsection{Proof of \Cref{thm:consistency_WR}.\ref{thm:consistency_WR_NN}}

\begin{proof}[Proof of \Cref{thm:consistency_WR}.\ref{thm:consistency_WR_NN}]
    Without loss of generality, let $t=0$ and $\sigma^\star=\sigma_0^\star$.
	Let $p = \hat p_W^{(n_0,n_1,\cC_\mathrm{nn})} = \frac{1}{n_0}\sum_{i\in\cN_0}w(Y_{\sigma^\star(i)}|Y_i)$ and $\bar p = \esp{w(Y^{(X_i)}(1)|Y_i(0))}$.
	We have, where $X_{\cN_1}=(X_k)_{k\in\cN_1}$:
	\begin{align*}
		\bar p - p &= \underbrace{\esp{ \frac{1}{n_0}\sum_{i\in\cN_0} \set{w(Y_i^{(X_i)}(1)|Y_i(0)) - w(Y_{\sigma^\star(i)}|Y_i)}\Big|X_{\cN_1}}}_{A_1}\\
		&\quad+\underbrace{ \frac{1}{n_0}\sum_{i\in\cN_0} \set{w(Y_{\sigma^\star(i)}|Y_i) - \esp{w(Y_{\sigma^\star(i)}|Y_i)|i\in\cN_0,X_{\cN_1}}}}_{A_2}
	\end{align*}
	\textbf{Control of $A_2$.}
	The term $A_2$ is controlled by computing variance.
	Let 
	\[a_i=w(Y_{\sigma^\star(i)}|Y_i) - \esp{w(Y_{\sigma^\star(i)}|Y_i)|i\in\cN_0,X_{\cN_1}}\,,\]
	and note that we have $\esp{a_i|i\in\cN_0,X_{\cN_1}}=0$.
	Let  
	\[p_k=\proba{\sigma(i)=k|X_{\cN_1}}\] 
	be the (random) weights of the (random) Voronoi cells associated to elements of $X_{\cN_1}$.
	We have, where $i\ne j\in\cN_0$ are arbitrary (note that conditioned on $X_{\cN_1}$, $\cN_0$ is fixed):
	\begin{align*}
		\var(A_2|X_{\cN_1})&=\frac{\var(a_i|X_{\cN_1})}{n_0}\\
		&\quad + \frac{n_0-1}{n_0}\Big(\esp{a_ia_j|\sigma(i)=\sigma(j),X_{\cN_1}}\proba{\sigma(i)=\sigma(j)|X_{\cN_1}}\\
		&\quad + \esp{a_ia_j|\sigma(i)\ne\sigma(j),X_{\cN_1}}\proba{\sigma(i)\ne\sigma(j)|X_{\cN_1}}\Big)\\
		&\leq \frac{1}{n_0} + \proba{\sigma(i)=\sigma(j)|X_{\cN_1}} + \frac{n_0-1}{n_0}\esp{a_ia_j|\sigma(i)\ne\sigma(j),X_{\cN_1}}\proba{\sigma(i)\ne\sigma(j)|X_{\cN_1}}\,.
	\end{align*}
	Conditionnally on $X_{\cN_1}$, $\sigma(i)$ and $\sigma(j)$ are independent random variables that assign $k$ with probability $p_k$.
	Thus, we can prove that $a_i,a_j$ are negatively correlated conditionally on $\sigma(i)\ne\sigma(j)$:
	\begin{align*}
		\esp{a_ia_j\one_\set{\sigma(i)\ne\sigma(j)}|X_{\cN_1}}&=\sum_{k\ne\ell\in\cN_1}p_kp_\ell\esp{a_ia_j|X_{\cN_1},\sigma(i)=k,\sigma(j)=\ell}\\
		&=\sum_{k\ne\ell\in\cN_1}p_kp_\ell\esp{a_i|X_{\cN_1},\sigma(i)=k,\sigma(j)=\ell}\esp{a_j|X_{\cN_1},\sigma(i)=k,\sigma(j)=\ell}\\
		&\qquad\quad\text{since }a_i\indep a_j|X_{\cN_1},\sigma(i)=k,\sigma(j)=\ell\\
		&=\sum_{k\ne\ell\in\cN_1}p_kp_\ell\esp{a_i|X_{\cN_1},\sigma(i)=k}\esp{a_j|X_{\cN_1},\sigma(j)=\ell}\\
		&=\sum_{k,\ell\in\cN_1}p_kp_\ell\esp{a_i|X_{\cN_1},\sigma(i)=k}\esp{a_j|X_{\cN_1},\sigma(j)=\ell}-\sum_{k\in\cN_1}p_k^2\esp{a_i|X_{\cN_1},\sigma(i)=k}^2\\
		&\leq\left(\sum_{k\in\cN_1}p_k\esp{a_i|X_{\cN_1},\sigma(i)=k}\right)^2\\
		&=0\,.
	\end{align*}
	Using $\esp{a_ia_j\one_\set{\sigma(i)\ne\sigma(j)}|X_{\cN_1}}=\esp{a_ia_j|X_{\cN_1},\sigma(i)\ne\sigma(j)}\proba{\sigma(i)\ne\sigma(j)}$; we thus have that:
	\begin{equation*}
	    \esp{a_ia_j|X_{\cN_1},\sigma(i)\ne\sigma(j)}\leq 0\,.
	\end{equation*}
	Thus, we have that $\var(A_2|X_{\cN_1})\leq \frac{1}{n_0} + \proba{\sigma(i)=\sigma(j)}$, and the last step of this first part of the proof is to show that $\proba{\sigma(i)=\sigma(j)}\to 0$, the purpose of the following lemma.
	\begin{lemma}\label{lem:sigma_equal}
		We have, for $i\ne j\in\cN_0$:
		\begin{equation*}
			\proba{\sigma(i)=\sigma(j)}\to 0\,.
		\end{equation*}
	\end{lemma}
	\begin{proof}[Proof of \Cref{lem:sigma_equal}]
		We have:
		\[\proba{\sigma(i)=\sigma(j)}=\frac{1}{n_1}\sum_{k\in\cN_1}\proba{\sigma(i)=k|\sigma(j)=k}\,.\]
		Let $k\in\cN_1$ and $x\in\mathrm{Supp}(X)$: $\forall \eps>0,\proba{X\in\cB(x,\eps)}>0$.

		\emph{First case:} $\proba{X=x}=p_x>0$. In that case, let $N_x=|\set{\ell\in\cN_1,X_\ell=x}$.
		We have that \[\proba{\sigma(i)=k|N_x,X_k=x,\sigma(j)=k}=\frac{1}{N_x}\,,\] and $N_x$ is a binomial random variable of parameters $(n_1,p_x)$.
		This leads to:
		\begin{align*}
			\proba{\sigma(i)=k|X_k=x,\sigma(j)=k}&=\sum_{N=1}^{n_1} 2^{-n_1}\frac{p_x^N(1-p_x)^{n_1-N}}{N}\binom{n_1}{N}\\
			&=\sum_{N=1}^{n_1} \frac{p_x^N(1-p_x)^{n_1-N}}{n_1+1}\binom{n_1+1}{N+1}\\
			&\leq\frac{1}{p_x(n_1+1)}\,.
		\end{align*}
		Thus, $\proba{\sigma(i)=k|X_k=x,\sigma(j)=k}\to 0$ as $n_1\to\infty$.

		\emph{Second case:} $\proba{X=x}=0$ (no Dirac mass).
		Let $\delta\in(0,1)$ and let $R>\eps>0$ such that $\proba{X\in\cB(x,\eps)}< \delta$ and $\proba{X\in\cB(x,R)}>1-\delta$.
		We cover $\cB(0,R)\setminus\cB(0,\eps)$ with $m$ balls of radius $\eps/2$: $\cB(0,R)\setminus\cB(0,\eps)\subset\bigcup_{r=1}^m\cB(z_r,\eps/2)$,
		where $z_r\in\cB(0,R)\setminus\cB(0,\eps)$.
		We remove all $z_r$ that satisfy $\proba{X\in\cB(z_r,\eps/2)}=0$ from this union.
		Let $\cE$ be the event $\set{\forall r\in[m],\exists \ell\in\cN_1\setminus\set{k},X_\ell\in\cB(z_r,\eps/2)}$.
		We have that 
		\begin{align*}
			\proba{\sigma(i)=k|X_k=x,\sigma(j)=k,\cE}&\leq \proba{X_i\in\cB(x,\eps)}\\
			&\leq \delta\,.
		\end{align*}
		Then,
		\begin{align*}
			\proba{\cE^C}\leq&\sum_{r=1}^m  \proba{\forall \ell\in\cN_1\setminus\set{k},X_\ell\notin\cB(z_r,\eps/2)}\\
			&\leq m(1-p_{\min})^{n_1-1}\,,
		\end{align*}
		where $p_{\min}=\min_{r\in[m]}\proba{X\in\cB(z_r,\eps/2)}$.
		Thus, $\proba{\cE}\to 1$, and $\proba{\sigma(i)=k|X_k=x,\sigma(j)=k}\leq 1-\proba{\cE} + \delta$.
		We can thus conclude that $\proba{\sigma(i)=k|X_k=x,\sigma(j)=k}\to 0$ as $n_1\to\infty$.

		\emph{Wrapping things up.}
		Using $\proba{\sigma(i)=k|\sigma(j)=k}=\int_\cX\proba{\sigma(i)=k|\sigma(j)=k,X_k=x}\dd\P(X_k=x)$, we have $\proba{\sigma(i)=k|\sigma(j)=k}\to 0$ as $n_1\to\infty$, using dominated convergence.
	\end{proof}
Using this, we have $\var(A_2)\to 0$, leading to $A_2\to0$ in probability.


\bigskip

\noindent \textbf{Control of $A_1$.}
We now control $A_1$, using the continuity assumption.
Using unconfoundedness:
	\begin{align*}
		&\qquad\qquad |A_1| \leq \esp{\delta(X_i,X_{\sigma^\star(i)},Y_i(0))|X_{\cN_1},T_i=0}\,,\qquad\text{where}\\
		& \delta(x,x',y)= \left|\esp{w(Y_i^{(X_i)}(1)|y)|X_i=x} - \esp{w(Y_j(1)|y)|X_j=x'}\right|\,.
	\end{align*}
	Let $\eps>0$ and $y$ fixed. Using our continuity and compactness assumptions, $x,x'\mapsto \delta(x,x',y)$ is uniformly continuous on $\cX\times \cX$, so that there exists $\eta>$ such that if $\NRM{x-x'}\leq \eta$, we have $\delta(x,x',y)\leq \eps$.
	We are going to show that with high probability, $\NRM{X_i-X_{\sigma^\star(i)}}\leq \eta$.
	Using compactness of $\cX$, there exist $u_1,\ldots,u_p\in\cX$ such that $\cX\subset\bigcup_{k=1}^p\cB(u_k,\eta/2)$.
	Let $p_k=\proba{X_i\in\cB(u_k,\eta/2)}$: we assume that $p_k>0$ for all $k$, otherwise we remove this ball and the corresponding $u_k$.
	Let $p_{\min}=\min_k p_k >0$. Let $k_x\in[p]$ such that $X_i\in\cB(u_{k_x},\eta/2)$. 
	We have, working conditionnally on $\cN_0,\cN_1,i\in\cN_0$:
	\begin{align*}
		\proba{\NRM{X_i-X_{\sigma^\star(i)}}>\eta}&\leq \esp{\proba{X_{\sigma^\star(i)}\notin\cB(u_{k_x},\eta/2) |k_x }}\\
		& = \esp{\proba{\forall j\in \cN_1\,,\,X_{j}\notin\cB(u_{k_x},\eta/2) |k_x }}\\
		& = \esp{(1-p_{k_x})^{n_1}}\\
		& \leq (1-p_{\min})^{n_1}\\
		& \underset{n_1\to \infty}{\longrightarrow} 0\,.
	\end{align*}
	This leads to:
	\begin{equation*}
		\proba{\delta(X_i,X_{\sigma(i)},y)>\eps}\leq (1-p_{\min})^{n_1}\,,
	\end{equation*}
	and thus $\proba{\delta(X_i,X_{\sigma(i)},y)\to 0}=1$ as $n_1\to\infty$, leading to $\esp{\delta(X_i,X_{\sigma(i)},Y_i(0))}\to 0$.
	We thus have that $\esp{|A_1|}\to 0$, and thus $A_1\to 0$ in probability, since $|A_1|\leq 1$ almost surely.
    This concludes the proof.
\end{proof}

\subsection{Proof of \Cref{thm:consistency_WR}.\ref{thm:consistency_WR_Complete}}

\begin{proof}[Proof of \Cref{thm:consistency_WR}.\ref{thm:consistency_WR_Complete}]
Let $\bar p=\esp{w(Y_i(1)|Y_j(0))}$ for $i\ne j$.
	We now prove the second point, with complete pairings.
Using our assumptions, we have that 
\[\esp{w(Y_i|Y_j)|T_i=1,T_j=0} = \esp{w(Y_j(1)|Y_i(0))}= \bar p\,,\]
so that 
\[\esp{\hat p_W^{(n_0,n_1,\cC_\mathrm{Tot})}}=\bar p\,.\]
Since $0\leq w \leq 1$, $\var(w(Y_j|Y_i))\leq 1$, leading to:
\begin{align*}
	\var\left(\hat p_W^{(n_0,n_1,\cC_\mathrm{Tot})}\right)&=\frac{1}{n_0^2n_1^2} \sum_{i,i'\in\cN_0,j,j'\in\cN_1}\esp{(w(Y_j|Y_i)-\bar p)(w(Y_{j'}|Y_{i'}))}\\
	&=\frac{1}{n_0^2n_1^2}\sum_{i\in\cN_0,j\in\cN_1}\esp{(w(Y_j|Y_i)-\bar p)^2}\\
	&\quad+ \underbrace{\frac{1}{n_0^2n_1^2} \sum_{i\ne i'\in\cN_0,j,j'\ne\cN_1}\esp{(w(Y_j|Y_i)-\bar p)(w(Y_{j'}|Y_{i'})-\bar p)}}_{=0\quad \text{(independence)}}\\
	&\quad + \frac{1}{n_0^2n_1^2} \sum_{i\in\cN_0,j\ne j'\in\cN_1}\esp{(w(Y_j|Y_i)-\bar p)(w(Y_{j'}|Y_{i'}))}\\
	&\quad + \frac{1}{n_0^2n_1^2} \sum_{i\ne i'\in\cN_0,j\in\cN_1}\esp{(w(Y_j|Y_i)-\bar p)(w(Y_{j'}|Y_{i'}))}\\
	&\leq \frac{1}{n_0n_1}+\frac{1}{n_0}+\frac{1}{n_1}\,.
\end{align*}
Thus, $\hat p_W^{(n_0,n_1,\cC_\mathrm{Tot})}\longrightarrow\bar p$ in probability as $n_0,n_1\to \infty$.
\end{proof}

\section{Proof: Nearest Neighbors and Observational Setting}

\begin{proof}[Proof of \Cref{thm:consistency_WR_observational}]
    The proof of the first point (consistency of $\hat\tau_\mathrm{NN}^{(t)}$) proceeds exactly as the proof of \Cref{thm:consistency_WR}/\Cref{thm:consistency_WR_NN}, except that the distribution of $X_i$ is not the same.
    The proof of the second points then combines this with $|\cN_t|/n\to \proba{T_i=t}$ almost surely.
\end{proof}

\section{Inverse propensity score weighted Nearest Neighbors}
\label{app:ipw}

through the following estimator:
\begin{equation}\label{eq:IPW_NN}
	\hat\tau_\mathrm{IPW} = \frac{1}{n}\sum_{i\in\cN_0}w(Y_{\sigma_1^\star(i)}|Y_i)(1-\hat\pi(X_i))^{-1}\,,
\end{equation}
where
\begin{equation*}
	\forall i\in\cN_1\,,\quad\sigma^\star_1(i)\in\argmin_{j\in\cN_1}\NRM{X_i-X_{j}}\,,
\end{equation*}
with uniform sampling if there are equalities.
Note that in the context of a RCT, propensities are known and are constant (for all $x\in\cX$, $\pi(x)=\pi$), and that $n_0/n$ is an unbiased and consistent estimate of $\pi$.
As such, \Cref{eq:estim_NN_def2} is simply a specific case of \Cref{eq:IPW_NN}.
We next show that $\hat \tau_\mathrm{IPW}$ is indeed a generalization of $p_\mathrm{W}$ with $\cC_\mathrm{NN}$ to observational data, since it is a consistent estimator of the same causal estimand.

\begin{theorem}
	Assume that \Cref{hyp:sutva,hyp:unconfoundedness,hyp:positivity} hold.
	Assume that $\hat \pi$ satisfies:
	\begin{enumerate}
	    \item Pointwise consistency almost surely: $\forall x\in\cX$, we have $\proba{\hat \pi(x)\to \pi(x)}=1$;
	    \item Mean consistency: $\esp{|\hat \pi(X_i)-\pi(X_i)|}\to 0$;
	    \item Finite and bounded second moment of propensity scores: $\limsup \esp{\frac{1}{(1-\hat\pi(X_i))^2}}$ and $\limsup \esp{\frac{1}{\hat\pi(X_i)^2}}<\infty$.
	\end{enumerate} 
	Assume finally that  $(x,y)\mapsto \esp{w(Y_i(1)|y)|X_i=x}$ is continuous in its first variable $x$ 		and that $\cX$ is compact.
	Then, $\hat\tau_\mathrm{IPW}$ (\Cref{eq:IPW_NN}) is a consistent estimator of $\tau_\star$ (\Cref{def:win_indiv}):
	\begin{equation*}
		\frac{1}{n}\sum_{i\in\cN_0}w(Y_{\sigma_1^\star(i)}|Y_i)(1-\hat\pi(X_i))^{-1}\underset{\P}{\longrightarrow} \tau_\star\,,
	\end{equation*}
	as $n_0,n_1\to\infty$.
\end{theorem}

As opposed to the inverse propensity weighting estimators provided by \citet{Mao2017} and by many subsequent works \citep{Chen2024,Chiaruttini2024,Guo2022,Yin2022,zhang2022causalinferencewinratio},
our estimator $\hat \tau_\mathrm{IPW}$ is consistent for $\tau_\star$ rather than for $\tau_\mathrm{pop}$.

\begin{proof}
	Let \[\hat p= \frac{1}{n}\sum_{i\in\cN_0}w(Y_{\sigma_1^\star(i)}|Y_i)(1-\hat\pi(X_i))^{-1}\]
	be the IPW estimator,
	\[\hat p^\star = \frac{1}{n}\sum_{i\in\cN_0}w(Y_{\sigma_1^\star(i)}|Y_i)(1-\pi(X_i))^{-1}\]
	be the IPW estimator with oracle propensities, and $p=\tau_\star$ be the targeted estimand.
	We have:
	\begin{equation*}
		\hat p-p=(\hat p-\hat p^\star)+(\hat p^\star -p)\,.
	\end{equation*}
	For this first term, 
	\begin{align*}
		|\hat p-\hat p^\star|&=\left|\frac{1}{n}\sum_{i=1}^n(1-T_i)w(Y_{\sigma_1^\star(i)}|Y_i)\set{(1-\hat\pi(X_i))^{-1}-(1-\pi(X_i))^{-1}}\right|\\
		&\leq \frac{1}{n}\sum_{i=1}^n(1-T_i)w(Y_{\sigma_1^\star(i)}|Y_i)\left|(1-\hat\pi(X_i))^{-1}-(1-\pi(X_i))^{-1}\right|\\
		&\leq \frac{1}{n}\sum_{i=1}^n\left|(1-\hat\pi(X_i))^{-1}-(1-\pi(X_i))^{-1}\right|\\
		&\leq \frac{1}{n}\sum_{i=1}^n\left\{\one_\set{\hat\pi(X_i)<1-\eta'}\left|(1-\hat\pi(X_i))^{-1}-(1-\pi(X_i))^{-1}\right|\right.\\
		&\qquad\qquad\left.+\one_\set{\hat\pi(X_i)\geq 1-\eta'}\left|(1-\hat\pi(X_i))^{-1}-(1-\pi(X_i))^{-1}\right|\right\}\,,
	\end{align*}
	for some given $0<\eta'<\eta$.
	First, notice that $u\mapsto (1-u)^{-2}$ is $\eta^{'-2}-$Lipschitz on $[0,1-\eta']$, so that:
	\begin{align*}
		\frac{1}{n}\sum_{i=1}^n\one_\set{\hat\pi(X_i)<1-\eta'}\left|(1-\hat\pi(X_i))^{-1}-(1-\pi(X_i))^{-1}\right|&\leq \frac{1}{n\eta^{'2}}\sum_{i=1}^n\one_\set{\hat\pi(X_i)<1-\eta'}\left|\hat\pi(X_i)-\pi(X_i)\right|\\
		&\leq \frac{1}{n\eta^{'2}}\sum_{i=1}^n\left|\hat\pi(X_i)-\pi(X_i)\right|\\
		&= \eta^{'-2}\esp{\left|\hat\pi(X_i)-\pi(X_i)\right|}\\
		&\quad+\frac{1}{n\eta^{'2}}\sum_{i=1}^n\set{\left|\hat\pi(X_i)-\pi(X_i)\right|-\esp{\left|\hat\pi(X_i)-\pi(X_i)\right|}}\,.
	\end{align*}
	Here, we have that $\esp{\left|\hat\pi(X_i)-\pi(X_i)\right|}\to0$ using mean consistency, while the second term converges to 0 in probability (sum of $n$ centered bounded random variables).
	Then, using Cauchy-Schwarz inequality, 
	\begin{align*}
		&\quad\frac{1}{n}\sum_{i=1}^n\one_\set{\hat\pi(X_i)\geq 1-\eta'}\left|(1-\hat\pi(X_i))^{-1}-(1-\pi(X_i))^{-1}\right|\\
		&\leq \sqrt{\frac{1}{n}\sum_{i=1}^n\one_\set{\hat\pi(X_i)\geq 1-\eta'} \times \frac{1}{n}\sum_{i=1}^n\left((1-\hat\pi(X_i))^{-1}-(1-\pi(X_i))^{-1}\right)^2}\,.
	\end{align*}
	The first factor in the square root satisfies 
	\begin{align*}
		\frac{1}{n}\sum_{i=1}^n\one_\set{\hat\pi(X_i)\geq 1-\eta'} &=\proba{ \hat\pi(X_i)\geq 1-\eta'} +\frac{1}{n}\sum_{i=1}^n\one_\set{\hat\pi(X_i)\geq 1-\eta'} -\proba{\hat\pi(X_i)\geq 1-\eta'}\,.
	\end{align*}
	The first term is deterministic and converges to zero almost surely thanks to pointwise convergence and dominated convergence, while the second term converges to zero as the averaged sum of $n$ independent, centered and bounded random variables.
	All this leads to $\hat p-\hat p^\star\to 0$ in probability and almost surely.

	For the second term $\hat p^\star -p$, we adapt the proof of \Cref{thm:consistency_WR}.\ref{thm:consistency_WR_NN}.
	We extend the defintion of $\sigma^\star$ to $\cN_1$: for $i\in\cN_1$ we have $\sigma^\star(i)=i$.
	We have, using that $p=\esp{w(Y_i(1),Y_i(0))}=\esp{\frac{1-T_i}{1-\pi(X_i)}w(Y_i(1),Y_i(0))}$ with unconfoundedness:
	\begin{align*}
		p - \hat p^\star &= \underbrace{\esp{ \frac{1}{n}\sum_{i\in\cN_0} \frac{w(Y^{(X_i)}(1)|Y_i(0)) - w(Y_{\sigma^\star(i)}|Y_i)}{1-\pi(X_i)}\Big|X_{\cN_1}}}_{A_1}\\
		&\quad-\underbrace{ \frac{1}{n}\sum_{i=1}^n \frac{w(Y_{\sigma^\star(i)}|Y_i)(1-T_i)}{1-\pi(X_i)}- \esp{\frac{w(Y_{\sigma^\star(i)}|Y_i)(1-T_i)}{1-\pi(X_i)}\Big|X_{\cN_1}}}_{A_2}\,.
	\end{align*}
	
	\noindent \textbf{Control of $A_2$.}
	The term $A_2$ is controlled computing its variance, as in the proof of \Cref{thm:consistency_WR}.\ref{thm:consistency_WR_NN}.
	Let $a_i=\frac{w(Y_{\sigma^\star(i)}|Y_i)(1-T_i)}{1-\pi(X_i)}- \esp{w(Y_{\sigma^\star(i)}|Y_i)|X_{\cN_1}}$.
	Since $\esp{a_i}=0$, we have $\esp{A_2|X_{\cN_1}}=0$.
	Then, using overlap, $|a_i|\leq 1/\eta$ almost surely, so that $\esp{a_i^2|\cN_1}\leq 1/\eta^2$.

	Let  $p_k=\proba{\sigma(i)=k|X_{\cN_1}}$ be the (random) weights of the (random) Voronoi cells associated to $X_{\cN_1}$.
	We have:
	\begin{align*}
		\var(A_2|X_{\cN_1})&=\frac{1}{n^2}\sum_{i=1}^n\var(a_i|X_{\cN_1})\\
		&\quad + \frac{1}{n^2}\sum_{i\ne j}\Big(\esp{a_ia_j|\sigma(i)=\sigma(j),X_{\cN_1}}\proba{\sigma(i)=\sigma(j)|X_{\cN_1}}\\
		&\qquad+ \esp{a_ia_j|\sigma(i)\ne\sigma(j),X_{\cN_1}}\proba{\sigma(i)\ne\sigma(j)|X_{\cN_1}}\Big)\\
		&\leq \frac{1}{\eta^2n} +\frac{1}{n^2}\sum_{i\ne j} \proba{\sigma(i)=\sigma(j)|X_{\cN_1}} + \frac{n-1}{n}\esp{a_ia_j|\sigma(i)\ne\sigma(j),X_{\cN_1}}\proba{\sigma(i)\ne\sigma(j)|X_{\cN_1}}
	\end{align*}
	Conditionnally on $X_{\cN_1}$ and on $i,j\in\cN_0$, $\sigma(i),\sigma(j)$ are independent random variables that assign $k$ with probability $p_k$.
	Thus, we can prove that $a_i,a_j$ are negatively correlated conditionally on $\sigma(i)\ne\sigma(j)$:
	\begin{align*}
		\esp{a_ia_j\one_\set{\sigma(i)\ne\sigma(j)}|X_{\cN_1}}&=\sum_{k\ne\ell\in\cN_1}p_kp_\ell\esp{a_ia_j|X_{\cN_1},\sigma(i)=k,\sigma(j)=\ell}\\
		&=\sum_{k\ne\ell\in\cN_1}p_kp_\ell\esp{a_i|X_{\cN_1},\sigma(i)=k,\sigma(j)=\ell}\esp{a_j|X_{\cN_1},\sigma(i)=k,\sigma(j)=\ell}\\
		&\qquad\quad\text{since }a_i\indep a_j|X_{\cN_1},\sigma(i)=k,\sigma(j)=\ell\\
		&=\sum_{k\ne\ell\in\cN_1}p_kp_\ell\esp{a_i|X_{\cN_1},\sigma(i)=k}\esp{a_j|X_{\cN_1},\sigma(j)=\ell}\\
		&=\sum_{k,\ell\in\cN_1}p_kp_\ell\esp{a_i|X_{\cN_1},\sigma(i)=k}\esp{a_j|X_{\cN_1},\sigma(j)=\ell}-\sum_{k\in\cN_1}p_k^2\esp{a_i|X_{\cN_1},\sigma(i)=k}^2\\
		&\leq\left(\sum_{k\in\cN_1}p_k\esp{a_i|X_{\cN_1},\sigma(i)=k}\right)^2\\
		&=0\,.
	\end{align*}
	Thus, we have that $\var(A_2)\leq \frac{1}{n_0} + \proba{\sigma(i)=\sigma(j)|i,j\in\cN_0}$ where $i\ne j$.
	\begin{lemma}\label{lem:sigma_equal}
		We have, for $i\ne j$:
		\begin{equation*}
			\proba{\sigma(i)=\sigma(j)|i,j\in\cN_0}\to 0\,.
		\end{equation*}
	\end{lemma}
	\begin{proof}[Proof of \Cref{lem:sigma_equal}]
		First, note that we have, for any measurable set $\cS\subset\cX$:
		\begin{equation*}
			\frac{\proba{X_i\in\cS|T_i=1}}{\proba{X_i\in\cS|T_i=0}}\in\left[\eta^2,\frac{1}{\eta^2}\right]\,.
		\end{equation*}
		Indeed, we have for $t=0,1$ and measurable set $\cS\subset\cX$ such that $\proba{X_i\in\cS}>0$:
		\begin{equation*}
			\proba{X_i\in\cS|T_i=t}=\frac{\proba{T_i=t|X_i\in\cS}}{\proba{T_i=t}}\times \proba{X_i\in\cS}\,.
		\end{equation*}
		Using overlap, this leads to $\frac{\proba{X_i\in\cS|T_i=1}}{\proba{X_i\in\cS|T_i=0}}\in\left[\eta^2,\frac{1}{\eta^2}\right]$.
		A consequence is that for all measurable $\cS\subset\cX$, we have $\proba{X\in\cS}=0\iff\proba{X\in\cS|T=0}=0\iff \proba{X\in\cS|T=1}=0$.

		We have $\proba{\sigma(i)=\sigma(j)}=\frac{1}{n_1}\sum_{k\in\cN_1}\proba{\sigma(i)=k|\sigma(j)=k}$.
		We work conditionally on $\cN_1,\cN_0$. Note that we have $n_0,n_1\to\infty$ almost surely as $n\to \infty$.
		Let $k\in\cN_1$ and $x\in\mathrm{Supp}(X)$: $\forall \eps>0,\proba{X\in\cB(x,\eps)}>0$.

		\emph{First case:} $\proba{X=x|T=1}=p_x>0$. In that case, let $N_x=|\set{\ell\in\cN_1,X_\ell=x}$.
		We have that 
		\[
			\proba{\sigma(i)=k|N_x,X_k=x,\sigma(j)=k}=\frac{1}{N_x}\,,
		\]
		and $N_x$ is a binomial random variable of parameters $(n_1,p_x)$.
		This leads to:
		\begin{align*}
			\proba{\sigma(i)=k|X_k=x,\sigma(j)=k}&=\sum_{N=1}^{n_1} 2^{-n_1}\frac{p_x^N(1-p_x)^{n_1-N}}{N}\binom{n_1}{N}\\
			&=\sum_{N=1}^{n_1} \frac{p_x^N(1-p_x)^{n_1-N}}{n_1+1}\binom{n_1+1}{N+1}\\
			&\leq\frac{1}{p_x(n_1+1)}\,.
		\end{align*}
		Thus, $\proba{\sigma(i)=k|X_k=x,\sigma(j)=k}\to 0$ as $n_1\to\infty$.

		\emph{Second case:} $\proba{X=x}=0$.
		Let $\delta\in(0,1)$.
		Let $R>\eps>0$ such that $\proba{X\in\cB(x,\eps)}< \delta$ and $\proba{X\in\cB(x,R)}>1-\delta$.
		We cover $\cB(0,R)\setminus\cB(0,\eps)$ with $m$ balls of radius $\eps/2$: $\cB(0,R)\setminus\cB(0,\eps)\subset\bigcup_{r=1}^m\cB(z_r,m)$,
		where $z_r\in\cB(0,R)\setminus\cB(0,\eps)$.
		We remove all $z_r$ that satisfy $\proba{X\in\cB(z_r,\eps/2)}=0$ from this union.
		Let $\cE$ be the event $\set{\forall r\in[m],\exists \ell\in\cN_1\setminus\set{k},X_\ell\in\cB(z_r,\eps/2)}$.
		We have that 
		\begin{align*}
			\proba{\sigma(i)=k|X_k=x,\sigma(j)=k,\cE}&\leq \proba{X_i\in\cB(x,\eps)}\\
			&\leq \delta/\eta^2\,.
		\end{align*}
		Then,
		\begin{align*}
			\proba{\cE^C}\leq&\sum_{r=1}^m  \proba{\forall \ell\in\cN_1\setminus\set{k},X_\ell\notin\cB(z_r,\eps/2)}\\
			&\leq m(1-p_{\min})^{n_1-1}\,,
		\end{align*}
		where $p_{\min}=\min_{r\in[m]}\proba{X\in\cB(z_r,\eps/2)|T=1}>0$.
		Thus, $\proba{\cE}\to 1$, and $\proba{\sigma(i)=k|X_k=x,\sigma(j)=k}\leq 1-\proba{\cE} + \delta$.
		We can thus conclude that $\proba{\sigma(i)=k|X_k=x,\sigma(j)=k}\to 0$ as $n_1\to\infty$.

		\emph{Wrapping things up.}
		Using $\proba{\sigma(i)=k|\sigma(j)=k}=\int_\cX\proba{\sigma(i)=k|\sigma(j)=k,X_k=x}\dd\P(X_k=x)$, we have $\proba{\sigma(i)=k|\sigma(j)=k}\to 0$, using dominated convergence.
	\end{proof}
Using this, we have $\var(A_2)\to 0$, leading to $A_2\to0$ in probability.


\noindent \textbf{Control of $A_1$.}
	Using unconfoundedness and overlap:
	\begin{align*}
		&\qquad\qquad |A_1| \leq \delta^{-1}\esp{\delta(X_i,X_{\sigma^\star(i)},Y_i(0))|X_{\cN_1},T_i=0}\,,\qquad\text{where}\\
		& \delta(x,x',y)= \left|\esp{w(Y_i(1)|y)|X_i=x} - \esp{w(Y_j(1)|y)|X_j=x'}\right|\,.
	\end{align*}
	Let $\eps>0$ and $y$ fixed. Using our continuity and compactness assumptions, $x,x'\mapsto \delta(x,x',y)$ is uniformly continuous on $\cX\times \cX$, so that there exists $\eta>$ such that if $\NRM{x-x'}\leq \eta$, we have $\delta(x,x',y)\leq \eps$.
	We are going to show that with high probability, $\NRM{X_i-X_{\sigma^\star(i)}}\leq \eta$.
	Using compactness of $\cX$, there exist $u_1,\ldots,u_p\in\cX$ such that $\cX\subset\bigcup_{k=1}^p\cB(u_k,\eta/2)$.
	Let $p_k=\proba{X_i\in\cB(u_k,\eta/2)}$: we assume that $p_k>0$ for all $k$, otherwise we remove this ball and the corresponding $u_k$.
	Let $p_{\min}=\min_k p_k >0$. Let $k_x\in[p]$ such that $X_i\in\cB(u_{k_x},\eta/2)$. 
	We have, working conditionnally on $\cN_0,\cN_1,i\in\cN_0$:
	\begin{align*}
		\proba{\NRM{X_i-X_{\sigma^\star(i)}}>\eta}&\leq \esp{\proba{X_{\sigma^\star(i)}\notin\cB(u_{k_x},\eta/2) |k_x }}\\
		& = \esp{\proba{\forall j\in \cN_1\,,\,X_{j}\notin\cB(u_{k_x},\eta/2) |k_x }}\\
		& = \esp{(1-\eta p_{k_x})^{n_1}}\\
		& \leq (1-\eta p_{\min})^{n_1}\\
		& \underset{n_1\to \infty}{\longrightarrow} 0\,.
	\end{align*}
	This leads to:
	\begin{equation*}
		\proba{\delta(X_i,X_{\sigma(i)},y)>\eps}\leq (1-\eta p_{\min})^{n_1}\,,
	\end{equation*}
	and thus $\proba{\delta(X_i,X_{\sigma(i)},y)\to 0}=1$ as $n_1\to\infty$, leading to $\esp{\delta(X_i,X_{\sigma(i)},Y_i(0))}\to 0$ using dominated convergence.
	We thus have that $\esp{|A_1|}\to 0$, and thus $A_1\to 0$ in probability, since $|A_1|\leq 1$ almost surely.
	
\end{proof}

\section{Proof: Consistency of Distributional Regression}

\begin{proof}[Proof of \Cref{thm:WR_DRF}]
	We have, for the estimator $\hat\tau$ defined in \Cref{eq:estim_without_IPW_DRF}:
	\begin{align*}
		\hat \tau-\tau_\star &= \frac{1}{n}\sum_{i=1}^n (1-T_i)\hat q_1(X_i,Y_i) + T_i\hat q_0(X_i,Y_i) - \esp{w(Y^{(X_i)}(1),Y_i(0))}\\
		&= \frac{1}{n}\sum_{i=1}^n (1-T_i) q_1(X_i,Y_i) + T_i q_0(X_i,Y_i) - \esp{w(Y^{(X_i)}(1),Y_i(0))}\\
		&\quad + \frac{1}{n}\sum_{i=1}^n (1-T_i)(\hat q_1(X_i,Y_i)-q_1(X_i,Y_i)) + T_i(\hat q_0(X_i,Y_i)- q_0(X_i,Y_i))\,.
	\end{align*}
	For the first term, we have that $(1-T_i) q_1(X_i,Y_i) + T_i q_0(X_i,Y_i)$ are \emph{i.i.d.} bounded random variables, of mean $\esp{w(Y^{(X_i)}(1),Y_i(0))}$,
	so that the first sum converges almost surely to 0.
	We even have, using the central limit theorem, that:
	\begin{equation*}
		\frac{1}{\sqrt n}\sum_{i=1}^n (1-T_i) q_1(X_i,Y_i) + T_i q_0(X_i,Y_i) - \esp{w(Y^{(X_i)}(1),Y_i(0))} \underset{\P}{\longrightarrow} \cN(0,\sigma_\infty^2)\,,
	\end{equation*}
	where $\sigma_\infty^2=\var\left((1-T_i) q_1(X_i,Y_i) + T_i q_0(X_i,Y_i) - \esp{w(Y^{(X_i)}(1),Y_i(0))}\right)$.
	For the second term, we have:
	\begin{align*}
		&\quad \left|\frac{1}{n}\sum_{i=1}^n (1-T_i)(\hat q_1(X_i,Y_i)-q_1(X_i,Y_i)) + T_i( q_0(X_i,Y_i)-\hat q_0(X_i,Y_i))\right|\\
		&\leq \frac{1}{n}\sum_{i=1}^n |\hat q_1(X_i,Y_i)-q_1(X_i,Y_i)| + |\hat q_0(X_i,Y_i)- q_0(X_i,Y_i)|\\
		&\leq \frac{1}{n}\sum_{i=1}^n (1-T_i)\big(|\hat q_1(X_i,Y_i)-q_1(X_i,Y_i)|-\esp{|\hat q_1(X_i,Y_i)-q_1(X_i,Y_i)||T_i=0}\big)\\
		&\qquad + T_i\big(|\hat q_0(X_i,Y_i)- q_0(X_i,Y_i)|-\esp{|\hat q_1(X_i,Y_i)-q_1(X_i,Y_i)|T_i=1}\big)\\
		&\quad + \esp{|\hat q_1(X_i,Y_i)-q_1(X_i,Y_i)||T_i=0} + \esp{|\hat q_0(X_i,Y_i)-q_0(X_i,Y_i)||T_i=1}\,.
	\end{align*}
	These last two terms are deterministic and converge to zero due to our assumptions.
	The big sum converges almost surely to zero, as the average of $n$ \emph{i.i.d.} centered and bounded random variables.
	This leads to the consistency of our estimator.
	For the asymptotic normality, it remains to prove that 
	\begin{align*}
		A_i&= \frac{1}{n}\sum_{i=1}^n (1-T_i)\big(|\hat q_1(X_i,Y_i)-q_1(X_i,Y_i)|-\esp{|\hat q_1(X_i,Y_i)-q_1(X_i,Y_i)||T_i=0}\big)\\
		&\qquad + T_i\big(|\hat q_0(X_i,Y_i)- q_0(X_i,Y_i)|-\esp{|\hat q_1(X_i,Y_i)-q_1(X_i,Y_i)|T_i=1}\big)
	\end{align*}
	is $o(1/\sqrt{n})$.
	We have 
	\begin{align*}
		&\quad\var\left((1-T_i)\big(|\hat q_1(X_i,Y_i)-q_1(X_i,Y_i)|-\esp{|\hat q_1(X_i,Y_i)-q_1(X_i,Y_i)||T_i=0}\big)\right)\\
		&\leq\var\left(|\hat q_1(X_i,Y_i)-q_1(X_i,Y_i)||T_i=0\right)\\
		&\leq\esp{|\hat q_1(X_i,Y_i)-q_1(X_i,Y_i)|^2|T_i=0}\\
		&\leq \esp{|\hat q_1(X_i,Y_i)-q_1(X_i,Y_i)||T_i=0}\,,
	\end{align*}
	since $|\hat q_1(X_i,Y_i)-q_1(X_i,Y_i)|\leq 1$ (these are probabilities). Thus, under ou assumption, the variance of each term of $A_i$ is $o(\sqrt{1/n})$.
	Then, $\proba{|1/n\sum_iA_i|>\eps/\sqrt{n}}\leq \var(A_i)\to 0$. This leads to $1/n\sum_iA_i|=o_\P(1/\sqrt{n})$, and concludes the proof.
\end{proof}

\end{document}